\documentclass[nohyperref]{article}

\usepackage{microtype}
\usepackage{graphicx}
\usepackage{subfigure}
\usepackage{booktabs} 

\usepackage{hyperref}

\usepackage[accepted]{icml2022}

\usepackage{amsmath}
\usepackage{amssymb}
\usepackage{mathtools}
\usepackage{amsthm}

\usepackage[capitalize,noabbrev]{cleveref}

\theoremstyle{plain}
\newtheorem{theorem}{Theorem}[section]

\newtheorem{lemma}[theorem]{Lemma}
\newtheorem{claim}[theorem]{Claim}
\newtheorem{corollary}[theorem]{Corollary}
\theoremstyle{definition}

\theoremstyle{remark}

\newtheorem{fact}[theorem]{Fact}

\usepackage[textsize=tiny]{todonotes}

\DeclareMathOperator*{\E}{\mathbb{E}}

\newcommand{\eps}{\varepsilon}

\newcommand{\supp}{\mbox{supp}}

\icmltitlerunning{Streaming Algorithms for Support-Aware Histograms}

\begin{document}

\twocolumn[
\icmltitle{Streaming Algorithms for Support-Aware Histograms}



\icmlsetsymbol{equal}{*}

\begin{icmlauthorlist}
\icmlauthor{Justin Y.\ Chen}{mit}
\icmlauthor{Piotr Indyk}{mit}
\icmlauthor{Tal Wagner}{msr}
\end{icmlauthorlist}

\icmlaffiliation{mit}{MIT, Cambridge, MA, USA}
\icmlaffiliation{msr}{Microsoft Research, Redmond, WA, USA}

\icmlcorrespondingauthor{Justin Chen}{justc@mit.edu}

\icmlkeywords{Machine Learning, ICML}

\vskip 0.3in
]



\printAffiliationsAndNotice{}  

\begin{abstract}
    Histograms, i.e., piece-wise constant approximations, are a popular tool used to represent data distributions. Traditionally, the difference between the histogram and the underlying distribution (i.e., the approximation error) is measured using the $L_p$ norm, which sums the differences between the two functions over all items in the domain. Although useful in many applications, the drawback of this error measure is that it treats approximation errors of all items in the same way, irrespective of whether the mass of an item is important for the downstream application that uses the approximation. As a result, even relatively simple distributions cannot be approximated by succinct histograms without incurring large error.
    
    In this paper, we address this issue by adapting the definition of approximation so that only the errors of the items that belong to the {\em support} of the distribution are considered. Under this definition, we develop efficient 1-pass and 2-pass streaming algorithms that compute near-optimal histograms in sub-linear space. We also present lower bounds on the space complexity of this problem. Surprisingly, under this notion of error, there is an exponential gap in the space complexity of 1-pass and 2-pass streaming algorithms. Finally, we demonstrate the utility of our algorithms on a collection of real and synthetic data sets. 
\end{abstract}

\section{Introduction}
The exponential growth of massive data sets over the recent years necessitated the development of methods for computing succinct summaries of data.  Such summaries should accurately preserve the desired data properties, while being succinct enough to be stored or communicated efficiently. A popular approach to representing  data succinctly is to compute their {\em histograms}. For a data set whose elements come from the universe $[n]=\{1 \ldots n\}$, a $k$-histogram approximation is a piece-wise constant function defined over $[n]$ consisting of $k$ pieces (see Section~\ref{sec:prelim} for a formal definition).  Thanks to its compact representation (which requires storing only $O(k)$ numbers), such histograms are often used to approximate the distribution of data attributes, and have applications to a variety of tasks such as approximate query processing, data visualization and density estimation.

Traditionally, the problem of finding a good histogram approximation to a given data set is formalized as follows. The algorithm is given  an empirical distribution $P$ over $[n]$, i.e., a sequence $p_1 \ldots p_n \ge 0$ such that $\sum_i p_i=1$. The goal is to find a $k$-histogram $f$ that minimizes the approximation error $\|P-f\|$, where $\|\cdot\|$ is some norm (typically $\ell_1$ or $\ell_2$).  Multiple exact and approximate algorithms for computing such histograms given input $P$ are known~\cite{cormode2012synopses}. The problem has been also studied  extensively in settings where the distribution $P$ is given {\em implicitly}. In particular, there has been a significant body of work in the streaming setting, where the algorithm has limited storage, and is allowed only one (or a small number of) passes over the data. Two streaming models have been considered: (i) the {\em time-series model}, where the probabilities $p_1, p_2, \ldots, p_n $ are given explicitly \cite{guha2001data,guha2004rehist,guha2005space,buragohain2007space, halim2009fast, terzi2006efficient}, or (ii) the {\em turnstile model}, where the algorithm is given a sequence of data elements $i \in [n]$, and each $p_i$ is defined implicitly as the fraction of times $i$ occurs in the input sequences \cite{gilbert2001surfing,gilbert2002fast,guha2002histogramming,muthukrishnan2005workload,cormode2006fast,hegde2016fast}. Other implicit input models, e.g., allowing the algorithm to sample elements from the distribution $P$, has been studied as well, see e.g., \cite{indyk2012approximating, chan2014optimal, chan2014efficient, acharya2015fast}, or a recent survey~\cite{canonne2020survey}. The aforementioned streaming results typically offer multiplicative error guarantees, while the sampling algorithms offer an additive error guarantee.

Although useful, the histogram approximation problem as formulated above suffers from a fundamental issue, which is that the approximation error is measured over {\em all} elements $i$ in $[n]$, regardless of whether the density of $i$ is important for the downstream application that uses the estimation.  This means that even very simple distributions cannot be approximated by $k$-histograms with sublinear values of $k$ without incurring essentially the maximum possible error. As an example, consider a distribution that is uniform over all {\em even} items $i \in [n]$. Such a distribution can be approximated by a $1$-histogram over even numbers with 0 error, but any approximation by an $o(n)$-histogram over the whole domain $[n]$ incurs error of $1-o(1)$, i.e., the maximum possible (Note that such error can be achieved by approximating all probabilities by $0$). See Figure~\ref{fig-example} for another example of this phenomenon.

\begin{figure}[t]
\begin{center}
\includegraphics[width=0.85\columnwidth]{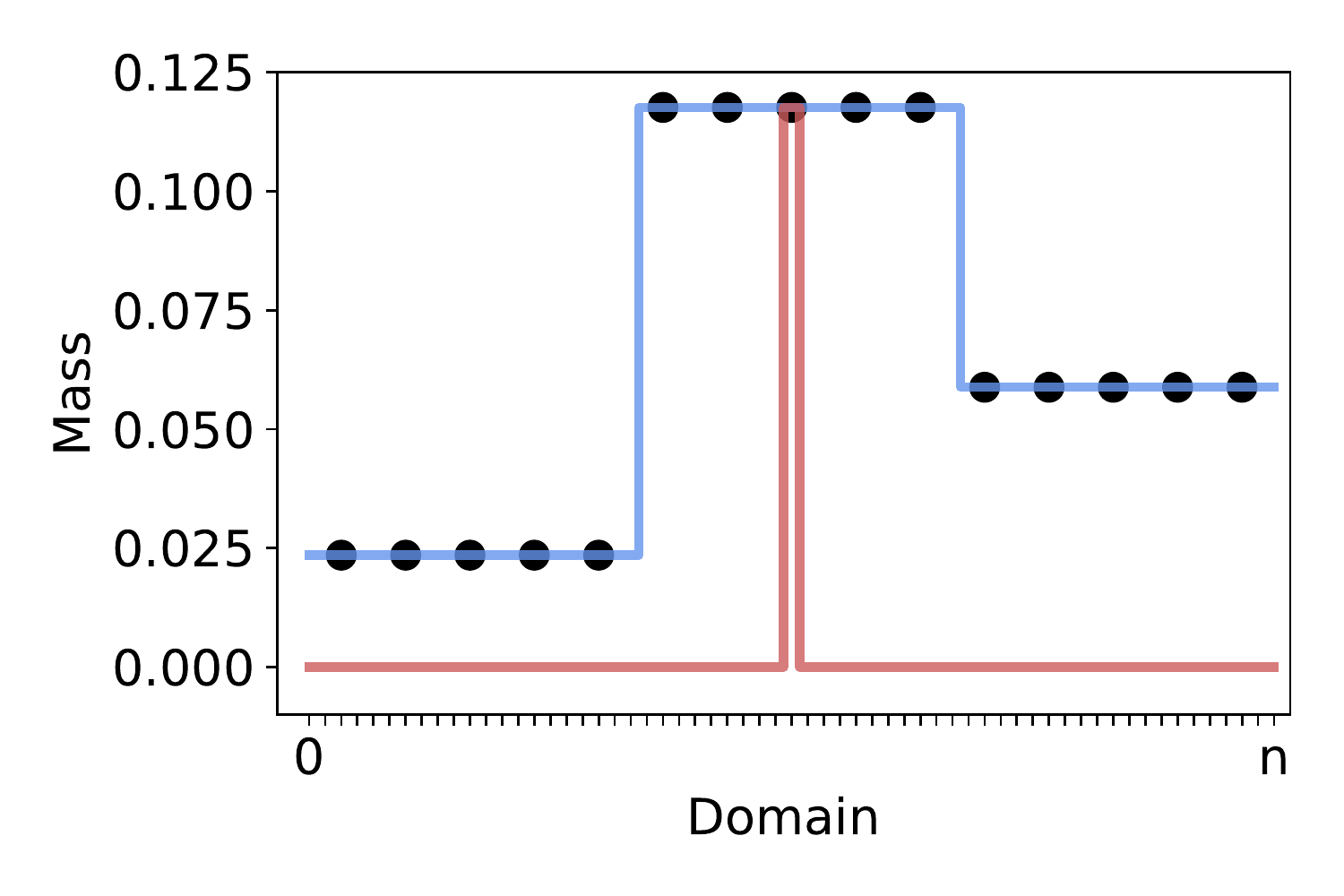}
\vspace{-1em}
\caption{Optimal support-aware (blue) and support-oblivious (red) 3-piece histograms on sparse data (black points). The support-aware histogram achieves zero support-aware error while the support-oblivious histogram only gives a non-trivial approximation on a single data point.}
\label{fig-example}
\end{center}
\end{figure}

In this paper, we propose to address this issue by making the following small but crucial modification to the definition of the problem: instead of measuring the approximating error over the whole domain $[n]$, we only measure it over the {\em support} of the distribution $P$. That is, for a given the distribution $P$ over the domain $[n]$, the goal is to compute a $k$-histogram $f$ over $[n]$ that (approximately) minimizes the error (see \Cref{sec:prelim} for a detailed definition):
\begin{equation}
\label{e:supp}
    err_P(f)= \|(P-f)_{\supp(P)}\|_s
\end{equation}

In the rest of the paper we will focus on the case $s=1$. We will refer to the error definition of Equation~\ref{e:supp} as  {\em support-aware L1 error}, and to histograms minimizing such error as {\em support-aware histograms}. By contrast, we will refer to the classic notion of measuring error over the entire domain as {\em support-oblivious} error.

The definition of support-aware error clearly avoids the simple example mentioned earlier. 
As shown in Section~\ref{sec:experiments}, it also  improves the quality of histogram approximation for real data. Thus, in this paper we focus on developing efficient streaming algorithms for approximately computing best support-aware histograms.

\subsection{Our results}

 In this paper we initiate the study of low-space streaming algorithms for support-aware histograms. Our main contributions are as follows. In all cases we consider randomized streaming algorithms with a constant probability of error, that operate in the so-called {\em strict turnstile} model. That is, the input stream can contain insertions as well as deletions of elements $i \in [n]$, as long as no element is deleted more often than inserted. For the multiset $S$ of the end of the stream (after accounting for all insertions and deletions), the distribution $P$ is defined by setting $p_i=m_i/m$, where $m_i$ is the number of occurrences of $i$ in $S$, and $m=\sum_i m_i$, where we assume that $m>0$. We use $H_k$ to denote the set of all $k$-histograms over $[n]$.

\begin{itemize}
\item We start from an observation that any  streaming algorithm  that simply detects whether $P$ can be approximated by a {\em single line} with zero error, i.e., whether $\min_{f \in H_1} err_P(f)=0$,  requires $\Omega(n)$ bits of space, even if the algorithm is allowed a constant number of passes. This stands in contrast to the aforementioned {\em support-oblivious} streaming algorithms which provide multiplicative error guarantees. Because of this limitation, we focus on additive error guarantees in the reminder of this paper.
\item We present two streaming algorithms  that compute support-aware histograms with additive error guarantees. That is, the algorithms compute histograms $f$ such that $err_P(f) \leq \min_{f^*\in H_k}err_P(f^*)+\epsilon$, for a parameter $\epsilon>0$. The first algorithm performs a single pass and uses $\sqrt{n}\cdot\log(n)\cdot k/\epsilon^3\cdot\mathrm{polylog}(k\log(n/\epsilon))$ space, while the second algorithm performs two passes and uses $\log^2(n) \cdot k/\epsilon^3 \cdot \mbox{polylog}(k, \epsilon^{-1})$ space. We note that the reported histogram $f$ might have more than $k$ pieces, but the number of pieces in $f$ does not exceed the space bound of the algorithm, i.e., the two algorithms have bi-criterion approximation guarantees.
\item We complement these results by a lower bound showing that any single pass algorithm reporting a histogram $f$ such that $err_P(f) \leq \min_{f^*\in H_k}err_P(f^*)+\epsilon$ must use $\Omega(\sqrt{n})$ bits of space, even when  $k=2$ and the algorithm is allowed to report a histogram with up to $o(\sqrt{n})$ pieces. This shows that the space bound of our 2-pass algorithm cannot be achieved in a single pass. 
\end{itemize}

On the empirical side, we analyze the performance of our one-pass and two-pass streaming algorithms on several real-world datasets, demonstrating the practical application of support-aware histograms. In particular, both of our algorithms achieve much lower (up to 3x) support-aware error than natural baselines, for the same amount of allocated space.

\subsection{Related work}
As we mention in the introduction, streaming and sampling algorithms for histogram estimation have been studied extensively. In addition to the works listed earlier,  the survey \cite{cormode2012synopses}, in particular section 3.7.1 ``Histograms Over Streaming Data'', provides a good (if somewhat dated) overview of the field. The vast majority of past work was focused on support-oblivious histograms. However, the following three papers are motivated by similar considerations as this paper, even though their technical development is quite different.

\begin{itemize}
    \item \cite{muthukrishnan2005workload} considered streaming algorithms for general ``workloads'', where for each data item $i$ the algorithm is given a weight $w_i$ and the goal is to minimize $\sum_i (p_i-f_i)^2 w_i$.  Unfortunately, their algorithms required storing essentially the whole weight vector $w_1 \ldots w_n$,  which required linear in $n$ space (unless the weight vectors can be compressed losslessly into smaller space). Our formulation can be viewed as a special case of their definition where the weight vector $w$ is the characteristic vector of the support of $P$. Since the support of $P$ defined implicitly by the input stream, there is no need to specify it or store it explicitly, circumventing the lower bounds for the general weights.
    \item \cite{batu2017generalized,diakonikolas2018sharp} considered sampling algorithms for detecting whether the input  distribution is uniform over its support, which essentially corresponds to checking whether it can be approximated with a 1-histogram. In their formulation, the optimal 1-histogram $f$  was the minimizer of  $\|(p-f_S)\|_1$, over all histograms $f$ {\em and} sets $S \subset [n]$, where $f_S$ is a vector of dimension $n$ such that $(f_S)_i=f_i$ if $i \in S$ and $(f_S)_i=0$ if $i \notin S$. Unfortunately, this definition is difficult to adopt in the streaming model, as storing the set $S$ requires linear space in general. Therefore, in our definition we essentially let $S= \supp(P)$ (note that  $S \subset \mbox{supp}(P)$ without loss of generality).
    \item \cite{du2021putting} considered density estimators with respect to the error measure as in Equation~\ref{e:supp}, and used them in the context of learning-augmented streaming algorithms. However, they represented the density function using neural networks, not histograms.  (Our paper was motivated in part by the goal of replacing neural networks with more computationally tractable density estimators that can be computed efficiently in the streaming model.)
\end{itemize}

\section{Preliminaries}\label{sec:prelim}
\paragraph{Setting: empirical distribution of a turnstile stream.}
As the input, we get a stream of insertions and deletions of elements in $[n]:=\{1,\ldots,n\}$, such that at any point in the stream, every item has been inserted at least as many times as it has been deleted (i.e., its count is non-negative). Let $m_i\geq0$ denote the count of $i\in[n]$ at the end of the stream, and let $m=\sum_{i=1}^nm_i$ be the total count. The empirical distribution $P=(p_1,\ldots,p_n)$ of the stream is given by $p_i=m_i/m$. Let $\supp(P)=\{i:p_i>0\}$ denote the support of $P$.
For simplicity, we assume that the length of the input stream is polynomial in $n$, which in particular implies that each $m_i$ can be represented using $O(\log n)$ bits.

\paragraph{Goal: support-aware histogram approximation.}
Our goal is to approximate $P$ by a histogram with few pieces. A \emph{$k$-piece histogram} is a function $f:[n]\rightarrow[0,1]$ with $i_1\leq\ldots\leq i_{k-1}\in [n]$ and $\gamma_1,\ldots,\gamma_k\in[0,1]$, such that, denoting $i_0=0$ and $i_k=n$, we have
\[ f(i)=\gamma_j \;\;\;\;\forall\;  j\in[k] \;\;\text{and}\;\; i\in\{i_{j-1}+1,\ldots,i_j\} . \]
In words, the delimiter indices $i_1\leq\ldots\leq i_{k-1}$ partition $[n]$ into $k$ pieces, and the histogram approximates each piece $j\in[k]$ with the fixed value $\gamma_j$. The \emph{support-aware $L1$ error} of approximating $P$ by $f$ is given, as in eq.~(\ref{e:supp}), by
\[ err_P(f) = \sum_{i\in\supp(P)}|p_i-f(i)| . \]

Let $H_k$ be the family of all $k$-histograms over $[n]$. Given an input error parameter $\epsilon\in[0,1]$ and number of pieces $k$, our goal is to find an approximately optimal histogram approximation for $P$, namely a histogram $f$ such that
\[ err_P(f) \leq \inf_{f^*\in H_k}err_P(f^*)+\epsilon. \]
In general, we allow a \emph{bicriteria approximation guarantee}, which means that the output histogram $f$ is allowed to have $k'\geq k$ pieces, with $k'$ being as close to $k$ as possible.


\paragraph{Sampling} One of the subroutines that our streaming algorithms use is  {\em $L0$ sampling}, which returns an item $i$ selected uniformly at random from $\supp(P)$, together with the value of $p_i$.
To accomplish the first task we use the one-pass streaming algorithm of~\cite{jowhari2011tight}  which returns one such samples with probability $1-\delta$ using $O(\log^2 n \log(1/\delta))$ space. 

\section{Lower bound for multiplicative approximation}

In this section we provide a simple lower bound showing that any streaming algorithm that computes a support-aware histogram with multiplicative error guarantees must use a linear amount of space, even when $k=1$ (i.e., the histogram is just a line) and even if the algorithm is allowed a constant number of passes. 
In fact, the lower bound holds even if the goal is to determine whether the support-aware error is equal to 0 or not. 

\begin{theorem}
Any streaming algorithm (in the insertions-only model, randomized with a constant probability of error) that determines whether $\min_{f \in H_1} err_P(f)=0$ requires $\Omega(n)$ bits of space, even if the algorithm is allowed to use $O(1)$ passes over the data. 
\end{theorem}

\begin{proof}
The proof follows via a reduction from Set Disjointness. Recall that Set Disjointness is a communication complexity problems involving two parties, Alice and Bob. The inputs of both parties are $n$-length bit vectors, denoted as $a_1 \ldots a_n$ and $b_1 \ldots b_n$, respectively. The goal of the problem is for both parties to determine whether there exists $i \in [n]$ such that $a_i=b_i=1$. It is known that this task requires that Alice and Bob exchange $\Omega(n)$ bits~\cite{kalyanasundaram1992probabilistic}.

The reduction is as follows. For simplicity we consider one-pass algorithms; generalization to $O(1)$ passes is immediate. Suppose there exists an $S$-space streaming algorithm $A$ as in the theorem statement. Then we can use it to solve Set Disjointness in $S$ space as follows. First, Alice creates a stream $S$ consisting of all elements $i$ such that $a_i=1$. The stream is input to $A$, and its state is transmitted to Bob. Then, Bob creates a stream $S'$ consisting of elements $j$ such that $b_j=1$, which is input to the algorithm $A$. Let $P$ be the distribution induced by the concatenation of streams $S$ and $S'$. It can be seen that its non-zero coordinates are all equal (and thus the support-aware error of a 1-histogram is 0) if and only if the sets defined by vectors $a$ and $b$ are disjoint. This implies that $S$ must be at least linear in $n$.
\end{proof}

\section{One-Pass Algorithm}
In this section we prove nearly matching upper and lower bounds for the space usage of one-pass streaming algorithm for support-aware histogram approximation, with additive error. The space usage is proportional to $\sqrt{n}$, where $n$ is the domain size. In the next section we show that by allowing an additional pass, the dependence on $n$ can be improved exponentially, to $\mathrm{polylog}(n)$.

We begin by presenting the one-pass algorithm. It uses an algorithm for \emph{$L1$ heavy hitters} computation over the stream~\cite{cormode2005improved}, defined as follows. 

\begin{fact}[L1-heavy hitters]\label{thm:hh}
Let $\ell\geq1$ be an integer and $\epsilon\in(0,1)$. 
There is a one-pass streaming algorithm in the turnstile model that uses $O(\epsilon^{-1} \ell  \log n)$ words of space, and returns $Z\subset[n]$, such that $|Z|=O(\ell)$, and every $i\in[n]$ such that $p_i\geq1/\ell$ (i.e., $i$ is an $(1/\ell$)-heavy hitter) is included in $Z$. In addition, for every $i\in Z$, it outputs $z_i\in[0,1]$ such that $|z_i-p_i|\leq\epsilon/\ell$.
\end{fact}

The $L1$ heavy hitters are elements on which the histogram approximation could accrue large error individually. Our one-pass algorithm computes them to ensure each one is approximated with high precision, and uses uniform sampling to approximate the non-heavy elements in the support. The algorithm is specified in \Cref{alg:onepass}, and its guarantees are summarized in the next theorem, whose proof appears in Appendix~\ref{sec:1pass_ub_proof}.

\begin{theorem}\label{thm:ub1pass}
\Cref{alg:onepass} uses $s=O(\sqrt{n}\cdot\log(n)\cdot k/\epsilon^3)\cdot\mathrm{polylog}(k\log(n/\epsilon))$ space, and outputs a histogram $f$ with $s$ pieces such that $err_P(f) \leq \min_{f^*\in H_k}err_P(f^*)+\epsilon$. 
\end{theorem}

\begin{algorithm}[tb]
\caption{\label{alg:onepass} One-Pass Algorithm}
\begin{flushleft}
{\bfseries Input:} stream $S$, parameters $k, \eps$  \\
{\bfseries Before the pass on the stream:}
\vspace{-1em}
\begin{algorithmic}[1]
\STATE $s \gets O(\sqrt{n}\cdot\log(n)\cdot k/\epsilon^3)\cdot\mathrm{polylog}(k\log(n/\epsilon))$ 
\STATE Draw $s$ uniform i.i.d. samples from $[n]$, with replacement. Denote the set of samples $S$.
\end{algorithmic}

{\bfseries During the pass on the stream:}
\vspace{-1em}
\begin{algorithmic}[1]
\STATE Run the $L1$ heavy hitter procedure from \Cref{thm:hh}, with $\ell=\sqrt{n}/\epsilon^2$, and let $Z$ be the output.
\STATE Concurrently, maintain the exact counts of the elements in $S$.
\end{algorithmic}

{\bfseries After the pass on the stream:}
\vspace{-1em}
\begin{algorithmic}[1]
\STATE Approximate each element in $i\in Z$ with its own histogram piece, with value $z_i$ (notation from \Cref{thm:hh}).
\STATE Use dynamic programming  (see ~\cite{cormode2012synopses}) to compute the best $k$-piece histogram for the elements in $S\setminus Z$.
\end{algorithmic}
\end{flushleft}
\end{algorithm}

\emph{Remark.} We remark that if we are given any upper bound $n'$ on the support size (i.e., such that $\supp(P)\leq n'$), then the $\sqrt{n}$ term in \Cref{thm:ub1pass} can be improved to $\sqrt{n'}$, by using $L0$ sampling~\cite{jowhari2011tight} instead of the uniform samples in \Cref{alg:onepass} (with a logarthmic blowup). 

\paragraph{Lower bound.} The next theorem that the $\sqrt{n}$ term is \emph{necessary} for one-pass algorithms, even for histograms with $k=2$ pieces. Its proof is given in Appendix~\ref{sec:1pass_lb_proof}.

\begin{theorem}\label{thm:lb1pass}
Let $\epsilon>0$ be a sufficiently small constant. 
Any one-pass streaming algorithm that outputs a histogram $f$ such that $err_P(f) \leq \min_{f^*\in H_2}err_P(f^*)+\epsilon$, where $f$ is allowed to have as many as $O(\sqrt{n})$ pieces, must use $\Omega(\sqrt{n})$ space.
\end{theorem}

We emphasize that the space lower bound in the above theorem holds even when the output histogram $f$ is allowed to use as many as $O(\sqrt{n})$ pieces, while only having to approximately match the performance of the best $2$-piece histogram $f^*$. Therefore the lower bound holds not only for proper histogram approximation (in which $f$ would be restricted to using only $2$ pieces), but also to the bicriteria guarantee of the upper bound in \Cref{thm:ub1pass}, thus complementing it directly up to low-order $\mathrm{poly}(k,\epsilon^{-1},\log n)$ factors.

\section{Two-Pass Algorithm}

While $\sqrt{n}$ space complexity is necessary and sufficient for streaming algorithms that take one pass over the data, given a second pass over the data, a polylogarthmic dependence on $n$ suffices.

\paragraph{Hierarchical Heavy Hitters} Our two-pass algorithm uses prior work on hierarchical heavy hitters~\cite{cormode2008finding}. For the purposes of the algorithm, hierarchical heavy hitters are defined over a full binary tree with leaves $[n]$ (for simplicity, assume $n$ is a power of two) s.t.\ the subtree rooted at the $i$th node at height $\ell$ has leaves $[i \cdot 2^\ell + 1, (i+1) \cdot 2^\ell]$ for $\ell \in [0, \lg n]$ and $i \in [0, n/2^\ell - 1]$. 

For a node $h$ in this tree, let $T(h)$ denote the subtree rooted at $h$ and let $D(h)$ denote the set of leaves in $T(h)$, corresponding to the set of domain elements covered by $h$. For a given heaviness threshold $\phi$, hierarchical heavy hitters are defined in a bottom up fashion. For a node $h$, let $V(h) \subseteq T(h)$ be the set of nodes in the subtree rooted at $h$ (not including $h$) that are hierarchical heavy hitters. Then, $h$ is itself a hierarchical heavy hitter if the domain elements $D(h) \setminus \left(\cup_{v \in V(h)} D(v)\right)$ have mass exceeding $\phi$ (e.g., the elements contained in $h$ are heavy even after removing all the elements from $h$'s heavy descendants).

We will use the following additional notation. Let $S(h) = D(h) \setminus \left(\cup_{v \in V(h)} D(v)\right)$ be the set of domain elements that count towards whether $h$ is a hierarchical heavy hitter. Let $l(h)$ and $r(h)$ denote the left and right children of $h$, respectively. We will refer to hierarchical heavy hitters which are located at the leaves of the tree to be singleton heavy hitters (corresponding to the normal single-element notion of heavy hitters).

\begin{algorithm}[ht!]
\caption{\label{alg:twopass} Two-Pass Algorithm}
\begin{flushleft}
{\bfseries Input:} stream $S$, parameters $k, \eps$  \\
{\bfseries First pass:}
\begin{algorithmic}[1]
\STATE $T \gets$ Hierarchical Heavy Hitters of $S$ with heaviness threshold $\eps/2k$ using full binary tree over $[n]$ hierarchy via~\cite{cormode2008finding}, Section 5.2
\STATE $L \gets \emptyset$ \COMMENT{Set of disjoint, maximal contiguous intervals with less than $\eps/2k$ mass}
\STATE $H \gets \emptyset$ \COMMENT{Set of singleton elements of $[n]$ with at least $\eps/2k$ mass}
\FOR{$h \in T$}
    \IF{$|S(h)| = 1$}
    \STATE $H \gets H \cup \{S(h)\}$ 
    \ELSE
    \FOR{each maximal contiguous interval $[a,b] \in S(l(h))$}
        \STATE $L \gets L \cup \{[a,b]\}$ 
    \ENDFOR
    \FOR{each maximal contiguous interval $[a',b'] \in S(r(h))$}
        \STATE $L \gets L \cup \{[a',b']\}$ 
    \ENDFOR
    \ENDIF
    \STATE $L' \gets$ all maximal contiguous intervals of $[n]$ not covered by $H \cup L$
    \STATE $L \gets L \cup L'$
\ENDFOR
\end{algorithmic}

{\bfseries Second pass:}
\begin{algorithmic}[1]
\STATE Exactly count the frequency of each element in $H$
\STATE For each interval in $L$, use the algorithm of~\cite{jowhari2011tight} to take $O(\eps^{-2} \log(k/\eps))$ i.i.d. samples selected uniformly at random from the support, together with their masses, and calculate the median of the sample masses
\STATE Output histogram approximation where each element in $H$ is approximated by its true mass, and each element in an interval in $L$ is approximated by the approximate median of its interval
\end{algorithmic}
\end{flushleft}
\end{algorithm}

\begin{theorem}
\label{thm:twopass}
There is a $2$-pass streaming algorithm with constant failure probability that uses $O(k \log^2(n) /\eps^3) \cdot \mbox{polylog}(k, \eps^{-1})$ space and outputs $f \in H_{c k/\eps}$ for some constant $c$ such that $err_P(f) \leq \min_{f' \in H_k} err_P(f') + \eps$.
\end{theorem}

In essence, the algorithm calculates hierarchical heavy hitters in the first pass in order to produce a partitioning of the domain into a small number of singleton heavy hitters and contiguous intervals s.t.\ each of the intervals has mass at most $\eps/2k$. In the second pass, the algorithm exactly estimates each of the heavy hitters and approximates each interval by a sample median. As none of the intervals are too heavy, this median approximation suffices to produce an approximation almost as good as the optimal $k$-piece histogram.

In what follows we assume that the hierarchical heavy hitter algorithm~\cite{cormode2008finding} succeeded, which occurs with constant failure probability. Furthermore, a set of items is called {\em heavy} if its total mass is at least $\eps/2k$; it is called {\em light} otherwise.

\begin{lemma}
\label{lem:Hsize}
$|H| \leq \left\lceil \frac{2k}{\eps} \right\rceil$
\end{lemma}

\begin{proof}
As no element belongs to multiple hierarchical heavy hitters, there can be at most $2 k/\eps$ hierarchical heavy hitters, of which the elements of $H$ are a subset.
\end{proof}

\begin{lemma}
\label{lem:Lsize}
$|L| = O\left(\frac{k}{\eps}\right)$
\end{lemma}

\begin{proof}
By the same reasoning as the previous lemma, there are at most $2k/\eps$ non-singleton heavy hitters.
Let $h$ be one such hierarchical heavy hitter.
The number of intervals added to $L$ from $h$ will be the number of maximal contiguous intervals in $S(l(h))$ and in $S(r(h))$.
Note that the subtrees $T(l(h))$ and $T(r(h))$ are disjoint.
Without loss of generality, consider $S(l(h))$.

If $T(l(h)) \setminus \{h\}$ contains no hierarchical heavy hitters, it will contain a single maximal contiguous interval of all leaf nodes in the subtree rooted at $l(h)$.
Let $u \in T(l(h)) \setminus \{h\}$ be a hierarchical heavy hitter s.t.\ there is no other hierarchical heavy hitter $v \in T(l(h))$ s.t.\ $u \in T(v)$ (i.e.\ $u$ is not a descendant of another hierarchical heavy hitter in $T(l(h))$).
Call such a hierarchical heavy hitter $u$ a ``direct descendant'' of $l(h)$. 
The union of the leaves of the subtrees rooted at the direct descendants of $l(h)$ are exactly the leaves that are in the subtree $T(l(h))$ but not included in $S(l(h))$ (by the definition of hierarchical heavy hitters).
Therefore, each direct descendant can add at most one additional maximal contiguous interval as each direct descendant removes a single contiguous subinterval, creating at most one new discontinuity.

Each hierarchical heavy hitter is the direct descendant of at most one other hierarchical heavy hitter. For the sake of contradiction, consider distinct hierarchical heavy hitters $u, v, w$ s.t.\ $u$ is a direct descendant of both $v$ and $w$. This means that $u \in T(v)$ and $u \in T(w)$ but $v \notin T(w)$ and $w \notin T(v)$. By the binary tree structure of the hierarchy, this is impossible: for any two subtrees with nonempty intersection, one subtree must contain the other.
Therefore, for each hierarchical heavy hitter, one maximal contiguous interval is added for $l(h)$, one for $r(h)$, and one is added for being a direct descendant of another hierarchical heavy hitter for a total of at most $\left\lceil \frac{6k}{\eps} \right\rceil$ intervals.

It remains to count the number of elements added to $L$ via the set $L'$. By the argument above and Lemma~\ref{lem:Hsize}, there are $O(k/\eps)$ elements/intervals in $H \cup L$ in step 15 of the algorithm.
Therefore, $L'$ can have size at most $O(k/\eps)$ as the elements in $L \cup H$ can only partition $[n]$ into at most $|L| + |H| + 1$ pieces.
So, at the end of the first pass, $|L| = O(k/\eps)$, as required.
\end{proof}

\begin{lemma}
\label{lem:Llight}
Each interval in $L$ has mass less than $\eps/2k$.
\end{lemma}

\begin{proof}

Consider any such non-singleton hierarchical heavy hitter $h$.
Note that $S(l(h))$ and $S(r(h))$ must both have mass less than $\eps/2k$.
If both were heavy, $S(h)$ would be empty and if, without loss of generality, only $l(h)$ was heavy, then $S(h) = S(r(h))$ which must have mass less than $\eps/2k$.
Every interval added to $L$ via a non-singleton heavy hitter is a subinterval of $S(l(h))$ or $S(r(h))$ and thus is light.
The only other intervals added to $L$ are those that were not included in any hierarchical heavy hitters and thus are also light.
\end{proof}

\begin{lemma}
\label{lem:median}
Let $x_1 \leq x_2 \leq \ldots \leq  x_n$ be a sorted list of real numbers in $[0,1]$, and let $\beta = \sum_{i=1}^n x_i$ be the total mass of the $x_i$'s. Let $M^*$ be the median of the list and let $\hat{M}_s$ be the median of a random sample (w/ replacement) of size $s$. Let $\ell(M) = \sum_{i=1}^n |x_i - M|$ be the $\ell_1$ error of approximating the set by $M$.
Then, for $\eps \in [0,1]$,
\[
    \Pr(\ell(\hat{M}_s) - \ell(M^*) > \eps \beta) \leq 2e^{-\eps^2 s / 8}.
\]
\end{lemma}

The proof follows standard techniques and can be found in Appendix~\ref{appendix-2pass}. We are now ready to prove the main theorem.

\begin{proof}[Proof of Theorem~\ref{thm:twopass}]
First, we will argue that the algorithm is correct.
Using the algorithm of~\cite{cormode2008finding}, with constant failure probability, we will correctly find all hierarchical heavy hitters.
By Lemma~\ref{lem:Lsize}, after the first pass, we will have a partition of the domain $[n]$ into the sets $H$ and $L$ where each element in $H$ is heavy and each interval in $L$ is light.
In the second pass, as we exactly count each element in $H$, our approximation has no error on these elements.

Assume that the median approximation for each interval in $L$ is done exactly and call the resulting histogram approximation $f'$.
Consider an optimal $k$-piece histogram $f^* \in H_k$ and any interval $\ell \in L$.
Consider the case where $f^*$ uses a single piece to cover the interval $\ell$.
Note that out of all constant approximations of the interval $\ell$, the median mass of the nonzero elements will minimize the support-aware $L1$ error.
Therefore, limited to the interval $\ell$, the approximation error of $f'$ will be at most that of $f^*$.

Now, consider the case where $f^*$ uses multiple pieces to cover $\ell$.
Note that over any interval, $f'$ will have approximation error no greater than that of the approximation that always predicts zero mass as the median is the optimal constant approximation.
The zero approximation incurs error equal to the mass of the interval, which we know is at most $\eps/2k$.
Therefore, in this case, $f'$ incurs additional error over $f^*$ of at most $\eps/2k$.
As $f^*$ has $k$ total pieces, at most $k$ distinct intervals in $L$ can be covered by multiple pieces of $f^*$.
Combining the two cases, $err_P(f') - err_P(f^*) \leq \eps/2$.

It remains to account for the error in approximating the medians of the intervals in $L$.
By Lemma~\ref{lem:median}, for each interval, by taking $O(\eps^{-2}\log(k/\eps))$ samples, the excess error on each interval due to sampling is at most $\eps/2$ times the mass of the interval (even after union bounding over all $O(k/\eps)$ intervals).
As the total mass of all intervals is at most $1$, $err_P(f) - err_P(f') \leq \eps/2$.
In total,
\[
    err_P(f) - err_P(f^*) \leq \eps,
\]
the approximation has error at most $\eps$ more than the optimal.

Now, we will consider the space usage of the algorithm. Finding all $\eps/2k$ hierarchical heavy hitters with constant failure probability requires $O(\frac{k \log n}{\eps} \log \frac{k}{\eps})$ space (while not shown here, the dependence on the failure probability is logarithmic).
By Lemmas~\ref{lem:Hsize} and~\ref{lem:Lsize}, storing $H$ and $L$ from the first pass takes $O(\frac{k \log n}{\eps})$ space with the $\log n$ factor coming from the bits needed to store the indexes of the singleton heavy hitters or boundaries of the intervals.

In the second pass, exactly counting each singleton heavy hitter takes $O(\log n)$ bits per element in $H$ for total space of $O(\frac{k \log n}{\eps})$ space. 
As we need $O(\eps^{-2} \log(k/\eps))$ samples per interval in $L$, the total number of samples needed will be $O(\frac{k \log(k/\eps)}{\eps^3})$.
Using the algorithm of~\cite{jowhari2011tight}, we can select a uniform sample from $\supp(P)$ with failure probability $\delta$ in $O(\log^2 n \log(1/\delta))$ space.
Setting $\delta$ s.t.\ we can union bound over all samples, the total space due to sampling is $O\left(\frac{k \log^2(n) \log(k/\eps) \log(k/\eps^3 \log(k/\eps))}{\eps^3}\right)$, completing the analysis.
\end{proof}



\section{Experiments}
\label{sec:experiments}

\begin{table}[ht]
\begin{center}
\begin{small}
\begin{sc}
\begin{tabular}{cccc}
\toprule
Dataset & Domain & Support & Stream \\
\midrule
Taxi & 605K & 418K & 1.3M \\
Caida & 17M & 58K & 30M \\
War \& Peace & 18K & 1.9K & 53K \\
McDonalds & 18K & 2.2K & 14K \\
\bottomrule
\end{tabular}
\end{sc}
\end{small}
\end{center}
\caption{Domain size, support size, and stream length of each dataset used in experiments.}
\label{table-datasets}
\end{table}

\begin{figure}[ht]
\begin{center}
\includegraphics[width=0.9\columnwidth]{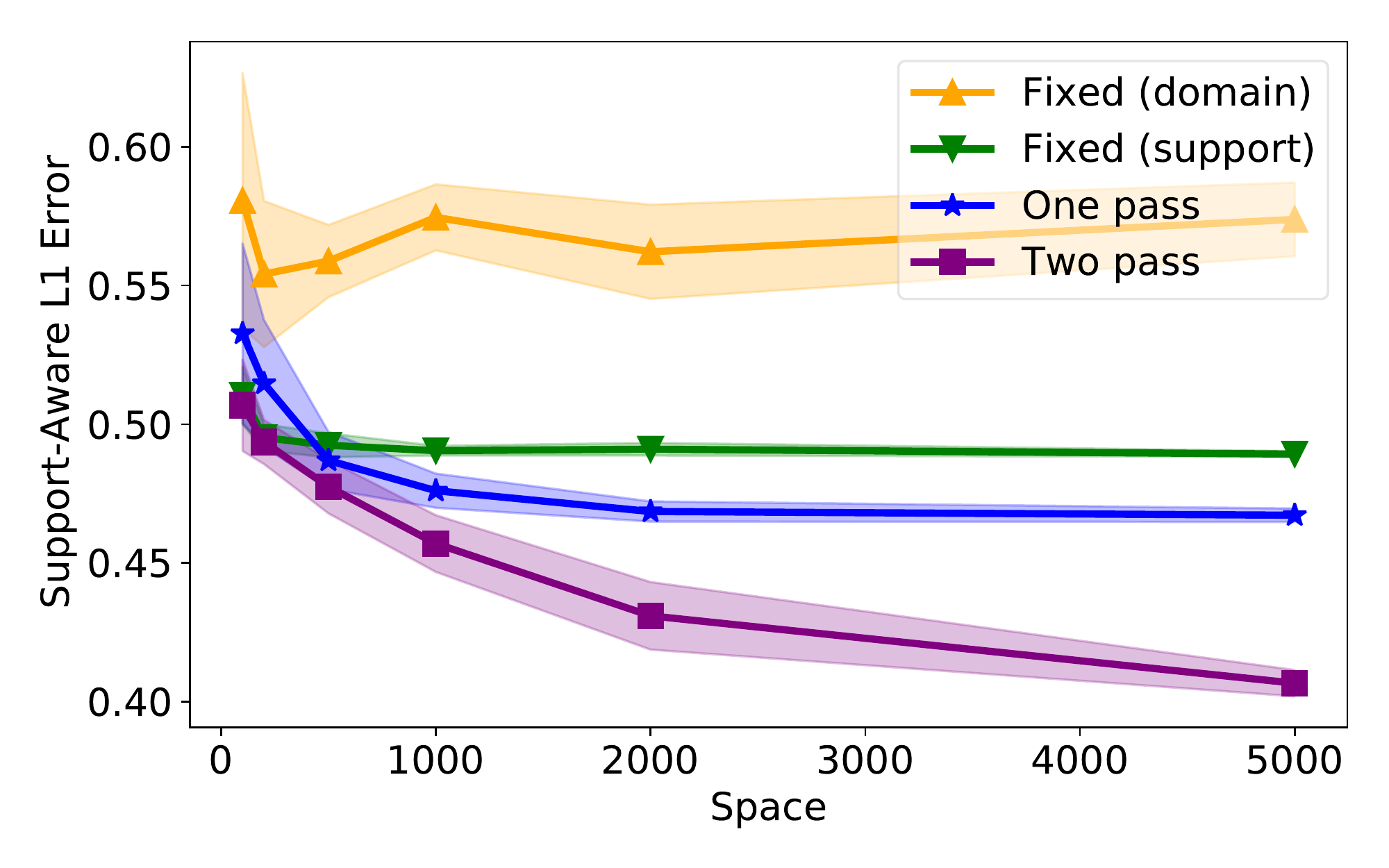}
\vspace{-1em}
\caption{Comparison of support-aware $L1$ error with varying space usage  with $k = 5$ on the Taxi dataset. Shading indicates one standard deviation over $10$ trials.}
\label{fig-taxi5}
\end{center}
\end{figure}

\begin{figure}[ht]
\begin{center}
\includegraphics[width=0.9\columnwidth]{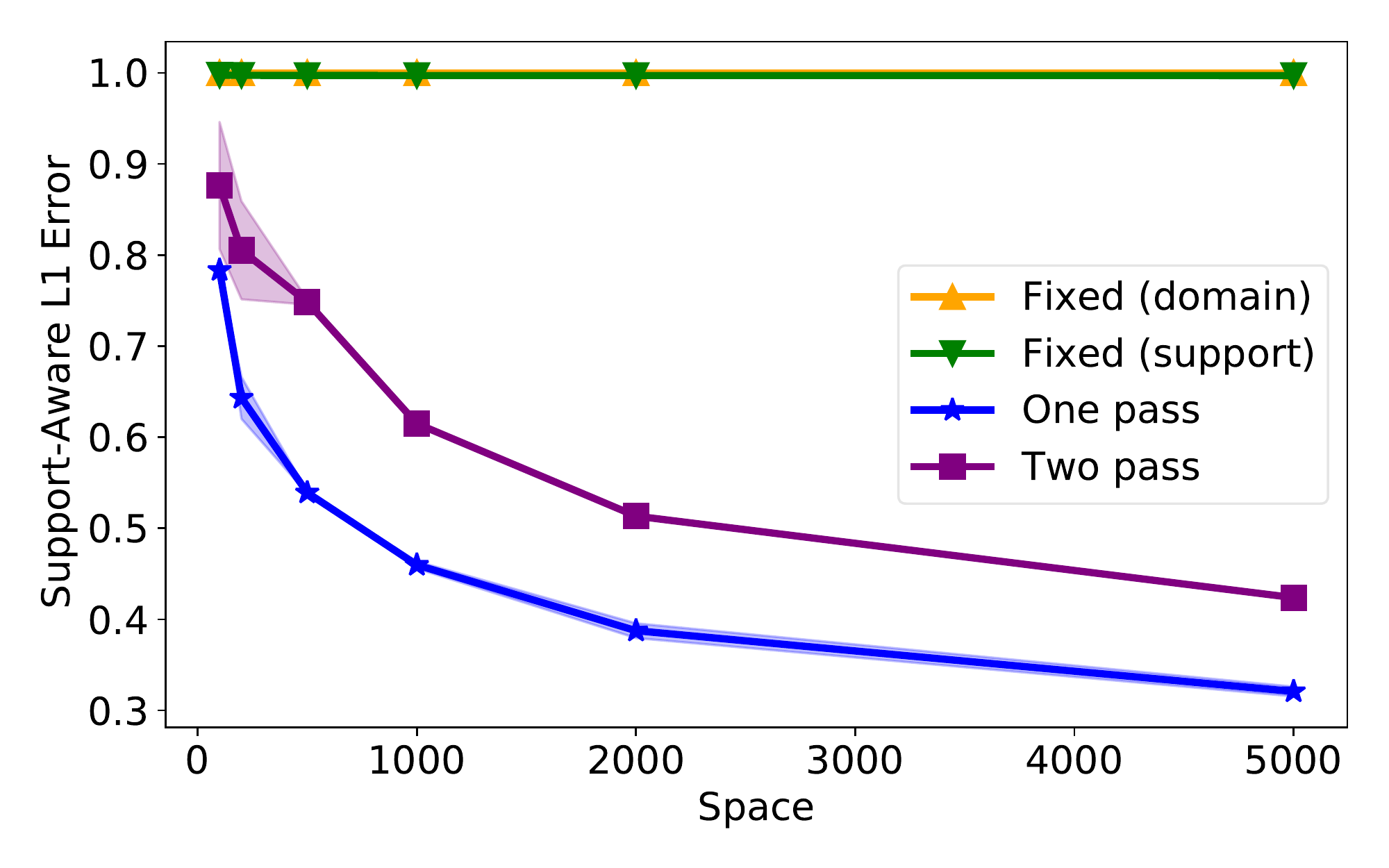}
\vspace{-1em}
\caption{Comparison of support-aware $L1$ error with varying space usage  with $k = 5$ on the CAIDA dataset. Shading indicates one standard deviation over $10$ trials.}
\label{fig-caida3-5}
\end{center}
\end{figure}

\begin{figure}[ht]
\begin{center}
\includegraphics[width=0.9\columnwidth]{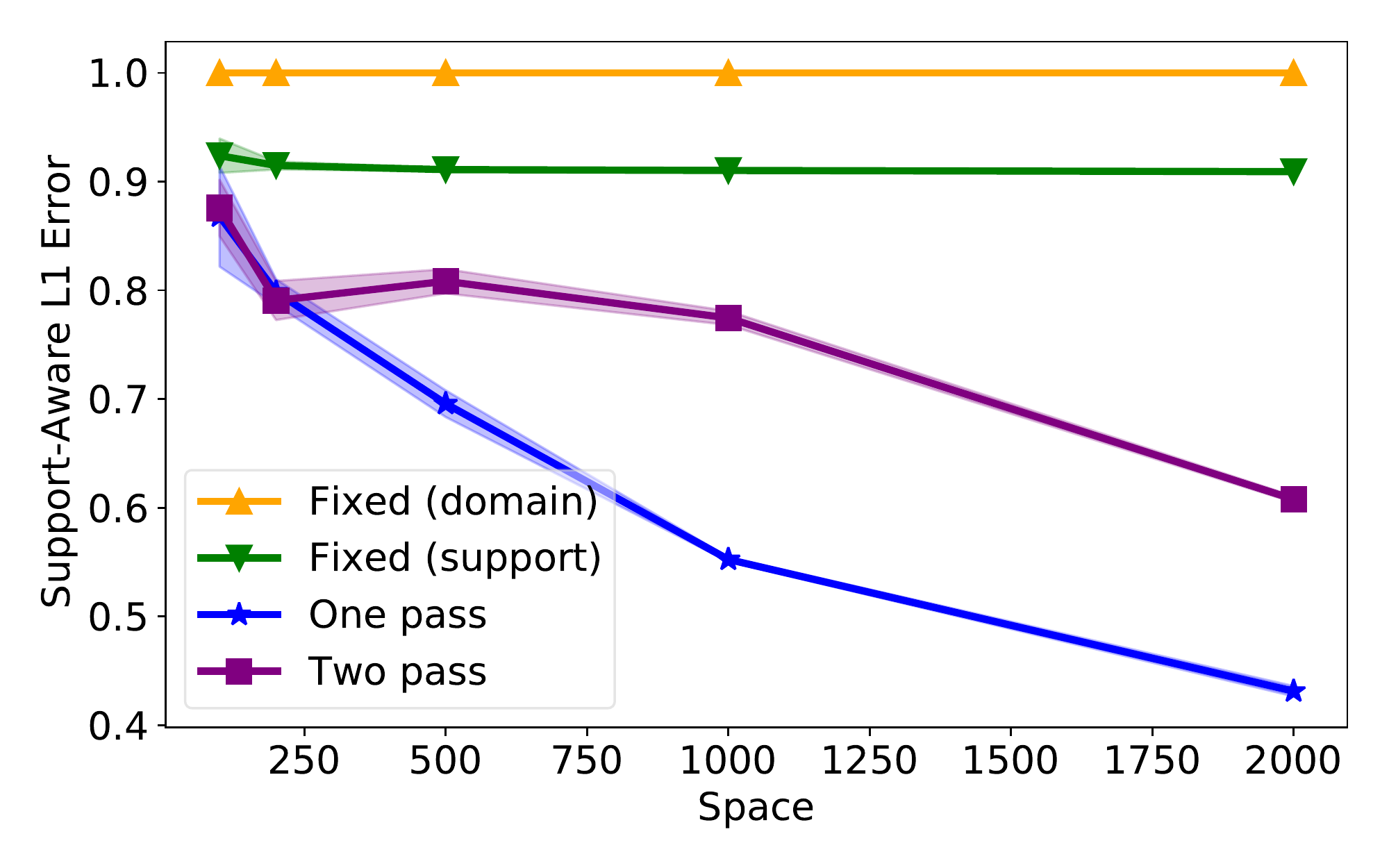}
\vspace{-1em}
\caption{Comparison of support-aware $L1$ error with varying space usage with $k = 5$ on the War \& Peace dataset. Shading indicates one standard deviation over $10$ trials.}
\label{fig-warpeace5}
\end{center}
\end{figure}

\begin{figure}[ht]
\begin{center}
\includegraphics[width=0.9\columnwidth]{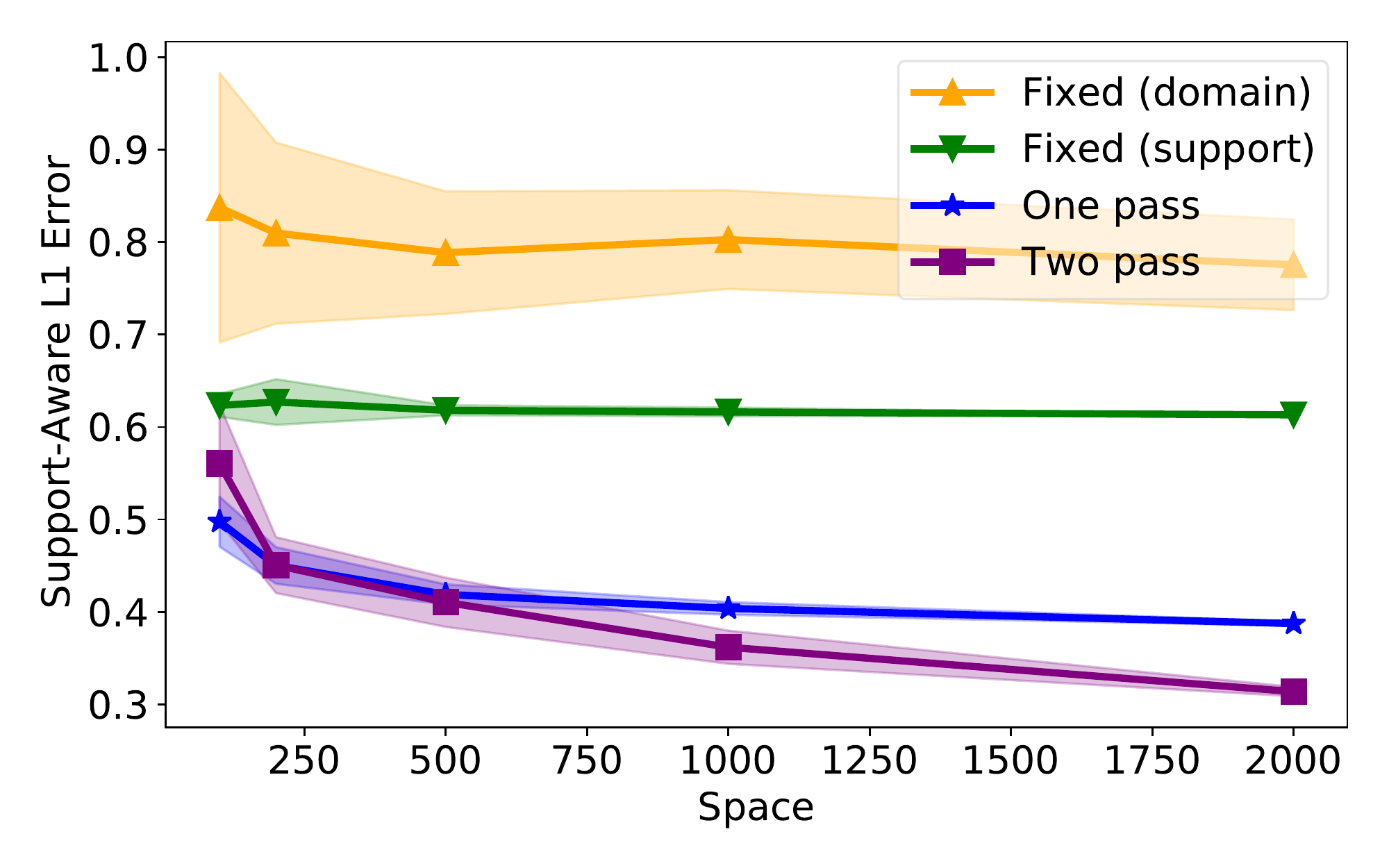}
\vspace{-2em}
\caption{Comparison of support-aware $L1$ error with varying space usage with $k = 5$ on the McDonalds dataset. Shading indicates one standard deviation over $10$ trials.}
\label{fig-mcdonalds5}
\end{center}
\end{figure}

We analyze the performance of our one-pass and two-pass streaming algorithms on a variety of real-world datasets, demonstrating the practical application of support-aware histograms.

\paragraph{Datasets} We compare histogram algorithms on four datasets. The Taxi dataset~\cite{nyc2021taxi} contains pickup times of yellow cabs in New York City in the month of January 2021. Stream elements come from (day of week, hour, minute, second) tuples with frequencies corresponding to the number of yellow cab pickups occurring at a given time throughout the month. The domain is in temporal order starting with midnight on Monday morning.

The CAIDA dataset~\cite{caida2016dataset} contains internet traffic data from a Tier1 ISP between Chicago and Seattle in 2016. Stream elements are the three most significant bytes of the destination IP addresses.

The War \& Peace dataset~\cite{tolstoy1869warandpeace} contains $3$-prefixes of all words in \emph{War and Peace} by Leo Tolstoy. Stream elements are three letter strings with frequencies corresponding to the number of times the strings start words in the novel. The domain is in lexicographic order.

The McDonalds dataset~\cite{mcdonalds} contains the locations of McDonalds restaurants in the United States and Canada. Stream elements are the latitudes of the restaurants, rounded to the nearest $0.01$.

The domain size, support size, and stream length of each dataset is displayed in Table~\ref{table-datasets}.

\paragraph{Baselines and Implementation}
We compare our one-pass and two-pass algorithms against two simple baselines. For given parameters $s, k$, both baselines split the domain into $k$ equal-sized intervals.
For each interval, the \emph{Fixed (support)} baseline takes $s/k$ $L0$ samples from the stream constrained to that interval and uses the median mass of those samples to approximate the interval.
The \emph{Fixed (domain)} baseline does the same but takes $s/k$ samples from the domain elements constrained to the interval rather than $L0$ sampling which only will sample non-zero domain elements. Hence, for sparse data, the Fixed (domain) baselines may often approximate intervals to have zero mass.
The Fixed (support) and Fixed (domain) baselines are natural algorithms that optimize for support-aware $L1$ error and support-oblivious $L1$ error, respectively.

We compare these baselines to our one-pass and two-pass algorithms (Algorithm~\ref{alg:onepass} and Algorithm~\ref{alg:twopass}).
The one-pass algorithm uses $s/2$ space to compute heavy hitters with approximate counts via the Space Saving algorithm~\cite{metwally2005spacesaving}. The rest of the $s/2$ space is used for $L0$ samples that form the set $S$ on which we will compute the best $k$-piece histogram.
Areas of the domain in between pieces of this histogram (as its endpoints are $L0$ samples, there are some parts of the domain not covered) are approximated by the median mass of all of the samples.

For the two-pass algorithm, as all of the streams in our experiments are insertion-only, we use the space-saving based implementation of hierarchical heavy hitters~\cite{mitzenmacher2012hhh}.
In the first pass, the heaviness threshold is set to $\frac{n \lg (n)}{s}$, and in the second pass, the $L0$ samples allowed by the budget $s$ are evenly distributed over all light intervals.

All of the baselines and algorithms involve storing sampled elements from the stream along with their masses. The parameter $s$ is varied to compare how the number of samples affects the performance of the algorithms (as shown in the x-axis of the figures). The $k$ parameter determines the number of pieces used by the baselines and by the one-pass algorithm (after removing heavy hitters). We set $k=5$ in the experiments in this section and display the same figures with $k=10$ in Appendix~\ref{appendix:experiments-k10}.
In Appendix~\ref{appendix:experiments-varyk}, we experiment with baselines which scale $k$ with the space budget to disambiguate the performance of our algorithms from the fact that they are improper and use more than $k$ pieces.
Each algorithm is run 10 times for each parameter setting. All figures display the mean support-aware $L1$ error along with one standard deviation shown in shading.

\paragraph{Results}
We start by just comparing the baselines.  For all but the CAIDA dataset (in which neither baseline gets non-trivial error), Fixed (support) significantly outperforms Fixed (domain) across space usage.
This is not surprising as Fixed (support) is optimized for support-aware error while Fixed (domain) is not. Still, these results show on real datasets that the sparsity of the support  makes support-aware error a meaningfully different error metric to the classic support-oblivious notion of error.
In particular, in these datasets, succinct histograms give better approximations for the non-zero elements rather than for the entire domain.

Compared to these baselines, our one-pass and two-pass algorithms achieve significantly smaller error, even with space usage in the hundreds (this is true for both the smaller War \& Peace and McDonalds datasets as well as the large-scale Taxi and CAIDA datasets). As space usage increases, the algorithms tend to produce better approximations, beating the baselines by up to 3x on the CAIDA and War \& Peace datasets. 

Intriguingly, on the CAIDA and War \& Peace datasets, the one-pass algorithm outperforms the two-pass algorithm though the two-pass algorithm has significantly smaller space complexity as indicated by our theoretical results.
These datasets contain many heavy hitters which contribute significantly to the error of the two-pass algorithm (or indeed of any histogram with few pieces).
While both the one-pass and two-pass algorithms compute heavy hitters from the stream and separately approximate them, the one-pass algorithm will store $s/2$ heavy hitters, while due to the higher heaviness threshold in the two-pass algorithm, only a fraction of that space will be used on singleton heavy hitters.
The two-pass algorithm will use the remaining space to fit medians to the intervals defined by the non-singleton hierarchical heavy hitters, allowing it to outperform the one-pass algorithm on the other datasets which have fewer outliers.

\section*{Acknowledgements}
This research was supported by the NSF TRIPODS program (award DMS-2022448), Simons Investigator Award, GIST-MIT Research Collaboration grant, MIT-IBM Watson collaboration, MathWorks Engineering Fellowship, and NSF Graduate Research Fellowship under Grant No. 1745302.

\bibliographystyle{icml2022}
\bibliography{ref}

\newpage
\appendix
\onecolumn

\section{Proof of one-pass algorithm}\label{sec:1pass_ub_proof}

In this section we prove \Cref{thm:ub1pass}.
We recall that the algorithm (\Cref{alg:onepass}) identifies the $L1$ heavy hitters and closely approximates them with their own histogram pieces, while approximating the rest of the elements using random sampling. Most of the proof is dedicated to dealing with the latter elements (the non-heavy hitters, whose mass is bounded by $\approx1/\sqrt{n}$). It will therefore be convenient to encompass this part of the proof in the following auxiliary theorem.

\begin{theorem}[bounded masses]\label{thm:boundedmass}
Let $\alpha\in(0,1)$ be constant. Suppose we are promised that $p_i\leq\epsilon^2/n^\alpha$ for all $i\in[n]$. 
Then, there is a 1-pass streaming algorithm that uses $\tilde O\left(n^{1-\alpha}\cdot k \cdot \log(n/\epsilon) \right)$ space, and outputs $f\in H_k$ such that $err_P(f) \leq \min_{f^*\in H_k}err_P(f^*)+\epsilon$.
\end{theorem}

\subsection{Proof of \Cref{thm:ub1pass}}
We first show how to obtain \Cref{thm:ub1pass} from \Cref{thm:boundedmass}. \Cref{alg:onepass} finds $L1$ heavy hitters whose mass is at least $\epsilon^2/\sqrt{n}$. 
More precisely, by \Cref{thm:hh}, it uses $O(\sqrt{n}/\epsilon^3)$ space to find a subset $Z\subset[n]$ of size at most $O(\sqrt{n}/\epsilon^2)$, such that every $i\in[n]$ with $p_i\geq\epsilon^2/\sqrt{n}$ belongs to $Z$. Furthermore, for each such $i$ it also returns $z_i$ such that $|z_i-p_i|\leq\epsilon^3/\sqrt{n}$. 
\Cref{alg:onepass} allots each $i\in Z$ its own histogram piece whose value is $z_i$, and thus accrues error of at most $\epsilon^3/\sqrt{n}$ on each $i\in Z$, hence a total error of at most $|Z|\cdot \epsilon^3/\sqrt{n} \leq \epsilon$ on all the elements in $Z$. The remaining elements $[n]\setminus Z$ have masses bounded by $\epsilon^2/\sqrt{n}$, and can be approximated by \Cref{thm:boundedmass} with $\alpha=1/2$, accruing an additional total error of $\epsilon$ of those elements. \Cref{thm:ub1pass} follows by scaling $\epsilon$ by a constant. 
As for the space usage of \Cref{alg:onepass}, \Cref{thm:hh} uses $O(\sqrt{n}\log(n)/\epsilon^3)$ space, and exactly counting the masses of the samples in $S$ takes $O(|S|)$ space, for a total of $O(\sqrt{n}\log(n)/\epsilon^3+|S|)=O(\epsilon^{-3}\cdot k\cdot \sqrt{n}\log(n)\cdot\mathrm{polylog}(k\log(n/\epsilon))$ space. \qed


\subsection{Proof of \Cref{thm:boundedmass}}
We now prove \Cref{thm:boundedmass}, thus completing the proof of \Cref{thm:ub1pass}. The algorithm described in the theorem is as follows. Before observing the stream, we draw $s =  O\left(n^{1-\alpha}\cdot k \cdot \log(n/\epsilon) \right)$ uniform i.i.d. samples from $[n]$ (with replacement). Denote the sequence of samples by $S$. We use our pass over the stream to count the exact mass $p_i$ of every $i$ in $S$. 

The space usage of the algorithm is clearly $O(s)$. The remainder of the proof will show that it achieves the guarantee of \Cref{thm:boundedmass}, albeit with error $O(\epsilon\cdot k\log(n/\epsilon))$ instead of $\epsilon$. Therefore, to get the requisite error guarantee, we scale $\epsilon$ down by $k\log(n/epsilon)$, bringing the total space usage to $O\left(n^{1-\alpha}\cdot k \cdot \left(\log(n/\epsilon)+\log\left(k\log (n/\epsilon)\right)\right) \right) = \tilde O(\epsilon\cdot k\log(n/\epsilon))$, as stated in the theorem.

In order to choose our output $f$ in the end of the pass on the stream, we need some more notation. Let $h\in H_k$ be a histogram. Recall it is defined by $i_1\leq\ldots\leq i_{i-1}\in [n]$ and $\gamma_1,\ldots,\gamma_k\in[0,1]$, with the convention $i_0=0$ and $i_k=n$. We split the items in each piece in $h$ into exponential intervals according to the masses. For $j=1,\ldots,k$ and $z=\lfloor\log_{1+\epsilon}(m)\rfloor,\ldots,\lceil\log_{1+\epsilon}(n^\alpha/\epsilon^2)\rceil$, define:
\[ I_{j,z}^{(h)} = \{i_{j-1}+1,\ldots,i_j\}\cap\{i:p_i\in\left((\tfrac1{1+\epsilon})^{z+1},(\tfrac1{1+\epsilon})^z\right]\}, \]
and 
\[ S_{j,z}^{(h)} = S\cap I_{j,z}^{(h)}  . \]
Now define our cost estimate for $h$ according to our samples $S$,
\[ est_S(h) := \sum_{j=1}^k\sum_z\frac{n}{s}\cdot|S_{j,z}^{(h)}|\cdot\left|\gamma_j-(\tfrac1{1+\epsilon})^z\right| . \]
As a small remark, note that $S$ here is treated as a sequence or multiset, so if the same element comes up twice in the sample, it is counted twice toward $|S_{j,z}^{(h)}|$.

Finally, for $\Gamma>0$, let $H_k^{[\Gamma]}$ denote the subset of $k$-histograms such that each of their values $\gamma_j$ is at most $\Gamma$, and such that each of their values $\gamma_j$ is an integer multiple of $\epsilon/n$. We will show that the histogram $f\in H_k^{[\epsilon^2/n^\alpha]}$ that minimizes the estimated error $est_S(f)$ is an approximately optimal solution to the original problem (i.e., it also minimizes the true error $err_P(\cdot)$ up to an additive error of at most $\epsilon$), and can be found by dynamic programming. 

We proceed to proving correctness, i.e., the approximate optimality of the returned histogram. The idea is that the random sample estimates the size of every sufficiently large $I_{j,z}^{(h)}$, while the smaller ones do not contribute much error anyway. Using those size estimates, and the fact that the histogram error on all points in a given interval is roughly the same, we can get a good estimate for its true cost.

More formally, fix $h\in H_k^{[\epsilon^2/n^\alpha]}$. As above, $h$ defines $I_{j,z}^{(h)}$, and and together with the sample it also defines $S_{j,z}^{(h)}$. Note that the contribution of an interval $I_{j,z}^{(h)}$ to the true error $err_P(h)$ is $\sum_{i\in I_{j,z}^{(h)}}|p_i-\gamma_j|$, while its contribution to the estimated error $est_S(h)$ is $\frac{n}{s}\cdot|S_{j,z}^{(h)}|\cdot\left|(\tfrac1{1+\epsilon})^z-\gamma_j\right|$. To argue about the relation between those quantities, we classify the intervals $I_{j,z}^{(h)}$ into ``heavy'' and ``light'' ones, where an interval is heavy if $|I_{j,z}^{(h)}|\geq10n^\alpha$, and light otherwise. Let $\mathcal I_H^{(h)}$ denote the set of heavy intervals, and $\mathcal I_L^{(h)}$ the light ones. 

For a heavy interval, we can show that its contributions to the true and estimated errors are both roughly the same, with probability high enough for a union bound over all candidate solutions.

\begin{lemma}[heavy intervals]\label{lmm:heavy}
If $|I_{j,z}^{(h)}|\geq10n^\alpha$, then with probability at least $1-(\epsilon/n)^{O(k)}$,
\[ \left| \sum_{i\in I_{j,z}^{(h)}}|p_i-\gamma_j| - \frac{n}{s}\cdot|S_{j,z}^{(h)}|\cdot\left|(\tfrac1{1+\epsilon})^z-\gamma_j\right| \right| \leq \epsilon^2 + \epsilon\sum_{i\in I_{j,z}^{(h)}}p_i . \]
\end{lemma}
\begin{proof}
Since $|I_{j,z}^{(h)}|\geq10n^\alpha$, by the Chernoff bound, our $s=O(n^{1-\alpha}\cdot k\log(n/\epsilon)$ samples suffice to have $||I_{j,z}^{(h)}|-\frac{n}{s}|S_{j,z}^{(h)}||\leq n^\alpha$ with probability $1-(\epsilon/n)^{O(k)}$. 
Furthermore, since for every $i\in I_{j,z}^{(h)}$ we have $p_i\in \left[(\tfrac1{1+\epsilon})^{z+1},(\tfrac1{1+\epsilon})^z\right)$, then
\[
  |p_i-(\tfrac1{1+\epsilon})^z| 
  \leq |(\tfrac1{1+\epsilon})^z - (\tfrac1{1+\epsilon})^{z+1}|
  = (1-\tfrac{1}{1+\epsilon})(\tfrac1{1+\epsilon})^z
  \leq (1-\tfrac{1}{1+\epsilon})(1+\epsilon)p_i
  = \epsilon\cdot p_i,
\]
and therefore,
\[
  |p_i-\gamma_j| = |p_i-(\tfrac1{1+\epsilon})^z| \pm |(\tfrac1{1+\epsilon})^z-\gamma_j| = |(\tfrac1{1+\epsilon})^z-\gamma_j| \pm \epsilon\cdot p_i.
\]
Together,
\begin{align*}
  \sum_{i\in I_{j,z}^{(h)}}|p_i-\gamma_j| &= \sum_{i\in I_{j,z}^{(h)}}\left(|(\tfrac1{1+\epsilon})^z-\gamma_j| \pm \epsilon\cdot p_i \right) \\
  &= |I_{j,z}^{(h)}|\cdot|(\tfrac1{1+\epsilon})^z-\gamma_j| \pm \epsilon\sum_{i\in I_{j,z}^{(h)}}p_i \\
  &= \left(\tfrac{n}{s}|S_{j,z}^{(h)}|\pm n^\alpha\right)\cdot|(\tfrac1{1+\epsilon})^z-\gamma_j| \pm \epsilon\sum_{i\in I_{j,z}^{(h)}}p_i .
\end{align*}
As a result, the quantity from the lemma statement that we are trying to upper-bound is at most
\[ n^\alpha \cdot|(\tfrac1{1+\epsilon})^z-\gamma_j| + \epsilon\sum_{i\in I_{j,z}^{(h)}}p_i . \]
For the first term, we recall that $\gamma_j\leq\epsilon^2/n^\alpha$ and $(\tfrac1{1+\epsilon})^z\leq\epsilon^2/n^\alpha$, thus $|(\tfrac1{1+\epsilon})^z-\gamma_j|\leq \epsilon^2/n^\alpha$ and the term is at most $\epsilon^2$.
\end{proof}

This implies the following corollary, that bounds the difference between the true and estimated errors, as long as they are measured only on the heavy intervals.
\begin{corollary}\label{cor:heavy}
With probability at least $0.999$, for all candidate solutions $h\in H_k^{[\epsilon^2/n^\alpha]}$ simultaneously, we have
\[  \left| \sum_{I_{j,z}^{(h)}\in\mathcal I_H^{(h)}}\sum_{i\in I_{j,z}^{(h)}}|p_i-\gamma_j| - \sum_{I_{j,z}^{(h)}\in\mathcal I_H^{(h)}}\sum_{i\in I_{j,z}^{(h)}}\frac{n}{s}\cdot|S_{j,z}^{(h)}|\cdot\left|(\tfrac1{1+\epsilon})^z-\gamma_j\right| \right| \leq  O(\epsilon\cdot k\log(n/\epsilon)) . \]
\end{corollary}
\begin{proof}
Recall that $H_k^{[\epsilon^2/n^\alpha]}$ is the set of $k$-histograms with values $\gamma_j$ which are multiple integers of $\epsilon/n$ in the range $[0,\epsilon^2/n^\alpha]$.
Thus, $\left|H_k^{[\epsilon^2/n^\alpha]}\right|=\binom{n}{k}\cdot(\tfrac n\epsilon)^k \leq (n/\epsilon)^{2k}$.
So in \cref{lmm:heavy} we can take a union bound over all heavy intervals induced by all $h\in H_k^{[\epsilon^2/n^\alpha]}$. As a result, the absolute difference between the two terms in the statement of the corollary is upper-bounded by 
\[  \epsilon^2|\mathcal I_H^{(h)}| + \epsilon \sum_{I_{j,z}^{(h)}\in\mathcal I_H^{(h)}}\sum_{i\in I_{j,z}^{(h)}}p_i .\]
The first term is at most $O(\epsilon k\log(n/\epsilon))$ since the number of intervals is $O(\epsilon^{-1}k\log(n/\epsilon))$. For the second term, note that the sum is over all $i$ that reside in heavy intervals, and we can upper-bound it by the sum over all items, $\epsilon\sum_{i=1}^np_i=\epsilon$.
\end{proof}

Next, light intervals. Their total \emph{true} contribution is small, just by being light. Their total \emph{estimated} contribution is also small, with more modest probability. Recall that $\mathcal I_L^{(h)}$ denotes the set of light intervals.

\begin{lemma}[light intervals]\label{lmm:light}
Fix any single $h\in H^{[\epsilon^2/n^\alpha]}_k$. 
With probability at least $0.999$,
\[ \sum_{I_{j,z}^{(h)}\in\mathcal I_L^{(h)}}\left| \sum_{i\in I_{j,z}^{(h)}}|p_i-\gamma_j| - \frac{n}{s}\cdot|S_{j,z}^{(h)}|\cdot\left|(\tfrac1{1+\epsilon})^z-\gamma_j\right| \right| \leq O(\epsilon \cdot k\log(n/\epsilon)). \]
\end{lemma}
\begin{proof}
Since a light interval $I_{j,z}^{(h)}\in\mathcal I_L^{(h)}$ satisfies $|I_{j,z}^{(h)}|\leq 10n^\alpha$, and since $|p_i-\gamma_j|\leq\epsilon^2/n^\alpha$ (regardless of the lightness of the interval), then 
\[ \sum_{I_{j,z}^{(h)}\in\mathcal I_L^{(h)}}\sum_{i\in I_{j,z}^{(h)}}|p_i-\gamma_j| \leq 10\epsilon^2\cdot|\mathcal I_L^{(h)}| . \]
Furthermore, since $\E[\frac{n}{s}|S_{j,z}^{(h)}|]=|I_{j,z}^{(h)}|\leq10n^\alpha$ and (deterministically) $\left|(\tfrac1{1+\epsilon})^z-\gamma_j\right|\leq\epsilon^2/n^\alpha$, 
\[
  \E\left[\sum_{I_{j,z}^{(h)}\in\mathcal I_L^{(h)}} \frac{n}{s}\cdot|S_{j,z}^{(h)}|\cdot\left|(\tfrac1{1+\epsilon})^z-\gamma_j\right| \right] \leq 10\epsilon^2\cdot|\mathcal I_L^{(h)}|  .
\]
By Markov's inequality, the probability the latter random variable does not exceed $10000\epsilon^2\cdot|\mathcal I_L^{(h)}|$ is at least $0.999$. The lemma is implied by noticing that there are at most $k\log_{(1+\epsilon)}(\epsilon^2/n^\alpha)=O(k\epsilon^{-1}\log(n/\epsilon))$ intervals, so $|\mathcal I_L^{(h)}| \leq O(k\epsilon^{-1}\log(n/\epsilon)$).
\end{proof}

Now we can prove the theorem with the following two claims.
\begin{claim}\label{clm:ub1aux}
With probability at least $0.999$, for all candidate solutions $h\in H_k^{[\epsilon^2/n^\alpha]}$ simultaneously, we have
\[  est_S(h) \geq err_P(h) - O(\epsilon\cdot k\log(n/\epsilon)) . \]
\end{claim}
\begin{proof}
On one hand,
\begin{align*}
  est_S(h) &= \sum_{I_{j,z}^{(h)}}\frac{n}{s}\cdot|S_{j,z}^{(h)}|\cdot\left|(\tfrac1{1+\epsilon})^z-\gamma_j\right| & \\
  &\geq \sum_{I_{j,z}^{(h)}\in\mathcal I_H^{(h)}}\sum_{i\in I_{j,z}^{(h)}}\frac{n}{s}\cdot|S_{j,z}^{(h)}|\cdot\left|(\tfrac1{1+\epsilon})^z-\gamma_j\right| & \text{restricting to heavy intervals} \\
  &\geq \sum_{I_{j,z}^{(h)}\in\mathcal I_H^{(h)}}\sum_{i\in I_{j,z}^{(h)}}|p_i-\gamma_j| - O(\epsilon\cdot k\log(n/\epsilon)) & \text{\Cref{cor:heavy}.}
\end{align*}
On the other hand,
\[
  err_P(h) = \sum_{I_{j,z}^{(h)}\in\mathcal I_H^{(h)}}\sum_{i\in I_{j,z}^{(h)}}|p_i-\gamma_j| + \sum_{I_{j,z}^{(h)}\in\mathcal I_L^{(h)}}\sum_{i\in I_{j,z}^{(h)}}|p_i-\gamma_j| ,
\]
and the second term (contribution of light intervals) was already upper-bounded by $10\epsilon^2|\mathcal I_L^{(h)}|=O(\epsilon k\log(n/\epsilon))$ in \Cref{lmm:light}, thus
\[ err_P(h) \leq \sum_{I_{j,z}^{(h)}\in\mathcal I_H^{(h)}}\sum_{i\in I_{j,z}^{(h)}}|p_i-\gamma_j| + O(\epsilon k\log(n/\epsilon)) , \]
and the claim follows.
\end{proof}

\begin{claim}\label{clm:ub2aux}
For the optimal solution $h^* \in H_k^{[\epsilon^2/n^\alpha]}$,
\[  est_S(h^*) \leq err_P(h^*) + O(\epsilon\cdot k\log(n/\epsilon)) . \]
\end{claim}
\begin{proof}
We break up the error contribution into heavy and light intervals as usual. \Cref{cor:heavy} tells us that the heavy contribution is the same in the true and estimated errors up to $\pm O(\epsilon\cdot k\log(n/\epsilon))$, for every $h\in H_k^{[\epsilon^2/n^\alpha]}$. \Cref{lmm:light} tells us that the light contribution is also the same up to $\pm O(\epsilon\cdot k\log(n/\epsilon))$ (the difference is that in the light lemma we only have enough probability to ensure this for any single $h\in h\in H_k^{[\epsilon^2/n^\alpha]}$, so we use it for the optimum $h^*$). Together we have
\[ |est_S(h^*) - err_P(h^*)| \leq O(\epsilon\cdot k\log(n/\epsilon)) , \]
which is stronger than the claim. 
\end{proof}

The two claims together show that returning $\hat h \in H_k^{[\epsilon^2/n^\alpha]}$ that minimizes the estimated error is almost as good as returning the one that minimizes the true error:
\begin{align*}
  err_P(\hat h) &\leq est_S(\hat h) + O(\epsilon\cdot k\log(n/\epsilon)) & \text{\Cref{clm:ub1aux}} \\
  &\leq est_S(h^*) + O(\epsilon\cdot k\log(n/\epsilon)) & \text{optimality of $\hat h$ w.r.t. estimated error} \\
  &\leq err_P(h^*) + O(\epsilon\cdot k\log(n/\epsilon)) & \text{\Cref{clm:ub2aux}.}
\end{align*}
So we return a solution with optimal value (in the discretized set of histograms $H_k^{[\epsilon^2/n^\alpha]}$) up to an additive loss of $O(\epsilon\cdot k\log(n/\epsilon))$. We can scale $\epsilon$ down by $O(k\log(n/\epsilon))$ as mentioned in the beginning of the proof. Finally it remains to observe that the discretization of $H_k^{[\epsilon^2/n^\alpha]}$ into integer multiples of $\epsilon/n$ doesn't matter since we lose only $\epsilon/n$ per $i\in[n]$ compared to non-discretized optimum, so only an additional $\epsilon$. \Cref{thm:boundedmass} is proven. \qed

\section{Proof of one-pass lower bound}\label{sec:1pass_lb_proof}
In this section we prove  \Cref{thm:lb1pass}. 
For clarity of presentation, we begin by proving the weaker version where the output histogram $f$ is allowed to use only $2$ pieces (\cref{thm:lbproper}). Afterwards, in \Cref{sec:fulllbproof}, we will show how to extend the proof and obtain the same lower bound for $f$ that can use as many as $O(\sqrt{n})$ pieces.

\begin{theorem}\label{thm:lbproper}
Let $\epsilon>0$ be a sufficiently small constant. 
Any one-pass streaming algorithm that outputs a $2$-piece histogram $f$ such that $err_P(f) \leq \min_{f^*\in H_2}err_P(f^*)+\epsilon$, must use $\Omega(\sqrt{n})$ space.
\end{theorem}
\begin{proof}
The proof is by reduction from the Augmented Indexing problem. We recall that Indexing is a one-way communication problem where Alice's input is a bitstring $a_1,...,a_t$, Bob's input is $j\in[t]$, Alice sends one message to Bob, and Bob needs to report  $a_j$ with probability better than $1/2$. In the Augmented Indexing variant, Bob also gets $a_1,\ldots,a_{j-1}$ as part of his input. Both variants are known to require $\Omega(t)$ communication. We set $t=\sqrt n$.

\paragraph{The reduction}
Our universe is $1,\ldots,3n$ where for simplicity $n$ is an integer square. The counts of all items $i=2n+1,\ldots,3n$ are always set to $1$. We partition $1,...,2n$ into $2\sqrt{n}$ equal-size consecutive intervals of length $\sqrt{n}$, and denote them as $A_1,B_1,A_2,B_2,\ldots,A_{\sqrt n}, B_{\sqrt n}$. We think of the $A_j$'s as Alice's cells and of the $B_j$'s as Bob's cells.

\paragraph{Alice's reduction:}
Given her input, for every $j=1,\ldots,\sqrt n$, Alice sets all elements in $A_j$ to $a_j$. So, every $A_j$ contains either no supported elements (if $a_j=0$) or $\sqrt n$ mice (if $a_j=1$). She does it by streaming the appropriate updates into the streaming algorithm, and then sends the memory state to Bob.

\paragraph{Bob's reduction:}
Bob streams into the memory state the following updates:
\begin{enumerate}
  \item He sets all elements in every $A_1,...,A_{j-1}$ to zeros (he knows where the nonzeros are by the Augmented part of Augmented Indexing).
  \item In each $B_1,\ldots,B_{j-1}$, he sets $\frac{\sqrt n}{j - 1}$ of the elements to have mass $\sqrt n$ (elephants). 
  \item In $B_j$, he sets $\gamma\sqrt{n}$ of the elements to have mass $\sqrt n$ (elephants), where $\gamma>\epsilon$ is a small constant.
\end{enumerate}

This concludes the stream. Note that the total count of masses is $O(n)$, the algorithm is guaranteed to return an optimal $2$-histogram up to an additive error of $\epsilon n$. 
Now Bob lets the algorithm find an approximately best $2$-histogram. Wlog, the values of the histogram are either $\sqrt n$ or $1$, since these are the only frequencies in the input. 
If at least half of the elements in $A_j$ are given histogram value $1$, Bob reports $1$, and reports $0$ otherwise.

To show the correctness of the communication protocol, first note that the algorithm has to put one piece with value $\sqrt n$ on the prefix (in order to cover the elephants in $B_1,\ldots,B_{j-1}$) and another piece with value $1$ on the suffix (to cover the mice in $2n+1\ldots,3n$), as otherwise the error is at least $\approx n$, which as we will see momentarily is much larger than the optimum in either case.
So the question is where it places the breakpoint between the two pieces.

Consider the case $a_j=0$. In this case, the optimal solution has zero error: we can cover $A_1,B_1,\ldots,A_j,B_j$ with the $\sqrt{n}$-piece, and the rest with the $1$-piece. 
So, the algorithm must move from the $\sqrt n$-piece to the $1$-piece \textbf{after} $B_j$ (and in particular $A_j$ needs to be covered by the $\sqrt n$-piece), since otherwise, it incurs error $\approx \gamma n$ on $B_j$, which is more than $\epsilon n$. 

Consider the case $a_j=1$. In this case, the optimal solution has error $\approx \gamma n$: we can cover $A_1,B_1,\ldots,A_{j-1},B_{j-1}$ with the $\sqrt n$-piece, and the rest with the $1$-piece, so that the only error we incur is $\approx \gamma n$ on $B_j$. 
So, the algorithm must move from the $\sqrt n$-piece to the $1$-piece \textbf{before} at least half of $A_j$ (and in particular, at least half of $A_j$ needs to be covered by the $1$-piece), since otherwise, it incurs error $\approx \tfrac12 n$ on $A_j$, and the gap from the optimal error $\approx \gamma n$ is more than $\approx \gamma n$ if $\gamma<\tfrac14$. 

So, as long as the algorithm succeeds with probability more than $0.5$, the theorem is proven.
\end{proof}

\subsection{Proof of \Cref{thm:lb1pass}}\label{sec:fulllbproof}
\begin{proof}

The proof is patterned by the above proof of \Cref{thm:lbproper}, with certain modifications to accommodate the much larger of pieces in the output histogram $f$.

\paragraph{The reduction}
Assume that we are given an algorithm that satisfies the guarantees in the theorem statement. Given an instance of Augmented Indexing over $\sqrt{n}$ bits, we will show how to solve the problem via the streaming algorithm on a stream over domain size $n$.

Assume for simplicity that $n$ is a perfect square.
We will split the domain of the stream into $\sqrt{n}$ contiguous chunks of size $\sqrt{n}$ called $S_1, \ldots, S_{\sqrt{n}}$.
For some constant $b \in [0,1]$ which we will later define, we will split each chunk $S_i$ into $b \sqrt{n}$ equal size subintervals $S_i^1, \ldots, S_i^{b\sqrt{n}}$. The first index of each subinterval will be reserved for Bob and the rest will be reserved for Alice.

\paragraph{Alice's reduction}
Alice will go through her $\sqrt{n}$ bits of her string $x$ and perform the following stream operations:
\begin{itemize}
    \item If $x_i = 0$, do nothing.
    \item If $x_i = 1$, add $a\sqrt{n}$ mice to chunk $S_i$ for some constant $a \in [0,1]$ in the following way. For each subinterval of $S_i$, pick $\frac{a}{b}\sqrt{n}$ indices other than the first index of the subinterval. Add a singleton element to the stream corresponding to each of these indices.
\end{itemize}
The result on the frequency distribution will be that for all of Alice's bits which equal $1$, there will be $a\sqrt{n}$ elements with mass $1/m$ if $m$ is the total mass of the stream.

\paragraph{Bob's reduction}
Bob has a single index $i$ of interest as well as knowledge of $x_1, \ldots, x_{i-1}$. Bob will execute the following stream updates.
\begin{itemize}
    \item For $j = 1,\ldots, i-1$, Bob will use his knowledge of Alice's bits to delete any elements Alice added in the chunks $S_1, \ldots S_{i-1}$.
    \item For each subinterval of $S_i$, Bob will add $\sqrt{n}$ copies of the first index of the subinterval.
\end{itemize}
The result on the frequency distribution will be that there will be zero mass on $S_1, \ldots S_{i-1}$ and there will be $b\sqrt{n}$ elements with mass $\sqrt{n}/m$ in chunk $S_i$.

After running the streaming algorithm, we will get some histogram approximation of the frequency distribution.
In order to use this to solve the indexing problem, we will do the following post-processing step.
Let $c \in [0,1]$ be some constant with $c < b$ and let the first indices of each of the subintervals of $S_i$ be refered to as as Bob's indices.
\begin{itemize}
    \item If at least $c\sqrt{n}$ of Bob's indices are approximated to have mass at least $\frac{1}{2\sqrt{n}}$, then report that $x_i = 1$.
    \item If at least $c\sqrt{n}$ of Bob's indices are approximated to have mass less than $\frac{1}{2\sqrt{n}}$, then report that $x_i = 0$.
\end{itemize}

Note that the total mass of the stream $m = \theta(n)$.
We will proceed by cases to show that if the streaming algorithm has the $\eps + opt_2(S)$ error guarantee in the theorem statement, then Bob will correctly recover Alice's $i$th bit.
For now, we will parameterize the number of pieces the streaming algorithm produces as $k\sqrt{n}$ for some constant $k$.

\paragraph{Case 1: $x_i = 0$.}
In this case, $opt_2(S) = 0$ as we can simply have the first piece of the histogram predict mass $\frac{1}{\sqrt{n}}$ for chunks $S_1,\ldots, S_i$ and predict mass $\frac{1}{n}$ for chunks $S_{i+1}, \ldots, S_{\sqrt{n}}$. As Alice and Bob's elements do not overlap, two pieces suffice to perfectly approximate the frequency distribution.

Assume that at least $c\sqrt{n}$ of Bob's indices are approximated to have mass less than $\frac{1}{2\sqrt{n}}$.
we will show that this implies that $\eps$ must be large.
In this case, the error incurred by the streaming algorithm will be
\[
    \left(\frac{1}{\sqrt{n}} - \frac{1}{2\sqrt{n}}\right) c\sqrt{n} = \frac{c}{2}.
\]
So, if $\eps < \frac{c}{2}$, then the reduction will correctly identify when $x_i = 0$ as many of Bob's indices must be predicted to have large mass. 
    
\paragraph{Case 2: $x_i = 1$.}
In this case, consider the two piece histogram that simply uses 1 piece and predicts mass $1/n$ everywhere. This gives an upper bound on $opt_2(S)$ (and for our setting of $a, b$ should actually be optimal).
\[
    opt_2(S) \leq \left(\frac{1}{\sqrt{n}} - \frac{1}{n}\right) b\sqrt{n} \leq b
\]

Assume that at least $c\sqrt{n}$ of Bob's indices are approximated to have mass at least $\frac{1}{2\sqrt{n}}$. We will show that this implies $\eps$ must be large.
Note that in this case, at least $(c - k)\sqrt{n}$ of Bob's indices must be covered by a histogram piece with height at least $\frac{1}{2\sqrt{n}}$ that also covers some other of Bob's indices (assuming $c > k$).
Thus, there must be $(c - k)\sqrt{n} \cdot \frac{a}{b}$ of Alice's elements which are predicted to have mass at least $\frac{1}{2\sqrt{n}}$.
So, the error of the streaming algorithm's approximation is at least
\[
    \left(\frac{1}{2\sqrt{n}} - \frac{1}{n}\right) \left(\frac{a}{b}\right)(c - k)\sqrt{n}
    = \frac{a(c - k)}{2b} - o(1).
\]
As we are guaranteed that the stream error is at most $opt_2(S) + \eps$, the reduction will correctly identify $x_i = 1$ as long as
\[
    \eps < \frac{a(c-k)}{2b} - b.
\]

It remains to give reasonable settings for the constants $a, b, c, k \in [0,1]$ under the following constraints
\begin{itemize}
    \item $a + b \leq 1$
    \item $a / b \in \mathbb{Z}$
    \item $b > c$
    \item $c > k$
    \item $\eps < \min\{\frac{c}{2}, \frac{a(c-k)}{2b} - b\}$
\end{itemize}

Setting $a = \frac{1}{2}, b = \frac{1}{8}, c = \frac{1}{10}, k = \frac{1}{40}$ satisfies all of the conditions.
Then, the bound on $\eps$ becomes
\[
    \eps < \min\{\frac{1}{20}, \frac{1}{40}\} = \frac{1}{40}.
\]

For this parameter regime, any streaming algorithm on domain size $n$ that outputs a histogram approximation with at most $\sqrt{n}/40$ pieces with error $opt_2(S) + \eps$ for $\eps < 40$ can be used to solve augmenting indexing on $\sqrt{n}$ bits and thus must use $\Omega(\sqrt{n})$ space.

\end{proof}

\section{Proof of Lemma~\ref{lem:median}}
\label{appendix-2pass}
\begin{proof}[Proof of Lemma~\ref{lem:median}]
Without loss of generality, assume $n$ is even.
Let $X_1, \ldots, X_s$ be random variables s.t.\ $X_i$ corresponds to the event that the $i$th random sample is less than $x_{n/2 - \delta n}$ for some $\delta \in [0,1/2]$.
Let $S = \sum_{i=1}^s X_i$ be the number of samples that are less than $x_{n/2 - \delta n}$.
$\hat{M}_s \leq x_{n/2 - \delta n}$ if and only if $S \geq s/2$. 
By Hoeffding bound,
\begin{eqnarray*}
    \Pr(\hat{M}_s \leq  x_{n/2 - \delta n}) & = &  \Pr(S \geq s/2)\\
    & = & \Pr(S \geq \E[S] + \delta s)
    \leq   e^{-2\delta^2 s}.
\end{eqnarray*}
The same bound holds for bounding the probability that $\hat{M}_s \geq x_{n/2 - \delta n}$.
So, with probability at least $1 - 2e^{-2 \delta^2 s}$, the sample median is within $\delta n$ of the rank of the true median.

Assume that $\hat{M}_s \in [x_{n/2 - \delta n}, x_{n/2 + \delta n}]$. Then, the loss of $M^*$ and $\hat{M}_s$ is equivalent for $i \in [1,n/2-\delta n] \cup [n/2 + \delta n, n]$:
\begin{align*}
    \sum_{i = 1}^{n/2 - \delta n}& |x_i - \hat{M}_s| + \sum_{i=n/2+\delta n}^n |x_i - \hat{M}_s| \\
    =& \sum_{i = 1}^{n/2 - \delta n} (\hat{M}_s - x_i) + \sum_{i=n/2+\delta n}^n (x_i - \hat{M}_s) \\
    =& \sum_{i = 1}^{n/2 - \delta n} (M^* - x_i - M^* + \hat{M}_s) \\
    &+ \sum_{i=n/2+\delta n}^n (x_i - M^* + M^* - \hat{M}_s) \\
    =& (n/2 - \delta n) (- M^* + \hat{M}_s + M^* - \hat{M}_s) \\
    &+ \sum_{i = 1}^{n/2 - \delta n} (M^* - x_i) + \sum_{i=n/2+\delta n}^n (x_i - M^*) \\
    =& \sum_{i = 1}^{n/2 - \delta n} |x_i - M^*| + \sum_{i=n/2+\delta n}^n |x_i - M^*|.
\end{align*}
The additional loss of $\hat{M}_s$ over $M^*$ on $i \in (n/2 - \delta n, n/2 + \delta n)$ is at most the mass of the elements in this range.
The mass of these elements is at most a $\frac{2\delta n}{n/2 + \delta n}$ fraction of the total mass.
Therefore,
\[
    \ell(\hat{M}_s) - \ell(M^*) \leq \frac{2\delta n}{n/2 + \delta n} \beta \leq 4 \delta \beta.
\]
Setting $\delta = \eps / 4$ completes the proof.
\end{proof}
\section{Additional Experiments: Varying $k$ with Space}
\label{appendix:experiments-varyk}
Both of our algorithms are improper--while their guarantees are in terms of $k$-piece histograms, they in fact use more than $k$ pieces. While this is standard in histogram approximation, from an empirical standpoint, it is important to understand whether the gains of our algorithms could be achieved by baselines which simply use more pieces. To that end, we present the same experiments as in Section~\ref{sec:experiments} but now with the ``Fixed'' baselines also using $k=Space/3$ or $k=Space/20$. Our algorithms as well as the solid line baselines use $k=5$.

For the Taxi dataset, these many-piece fixed algorithms indeed perform very well, but for the CAIDA, War \& Peace, and McDonalds datasets, our algorithms still outperform these baselines (for McDonalds, the baselines match the performance of our one pass algorithm). The Taxi dataset has no heavy hitters and is thus an easy case for the fixed baselines. On datasets with heavy hitters or less uniform structure, these baselines do not perform as well as our algorithms.

\begin{figure}[H]
\begin{center}
\includegraphics[width=0.4\columnwidth]{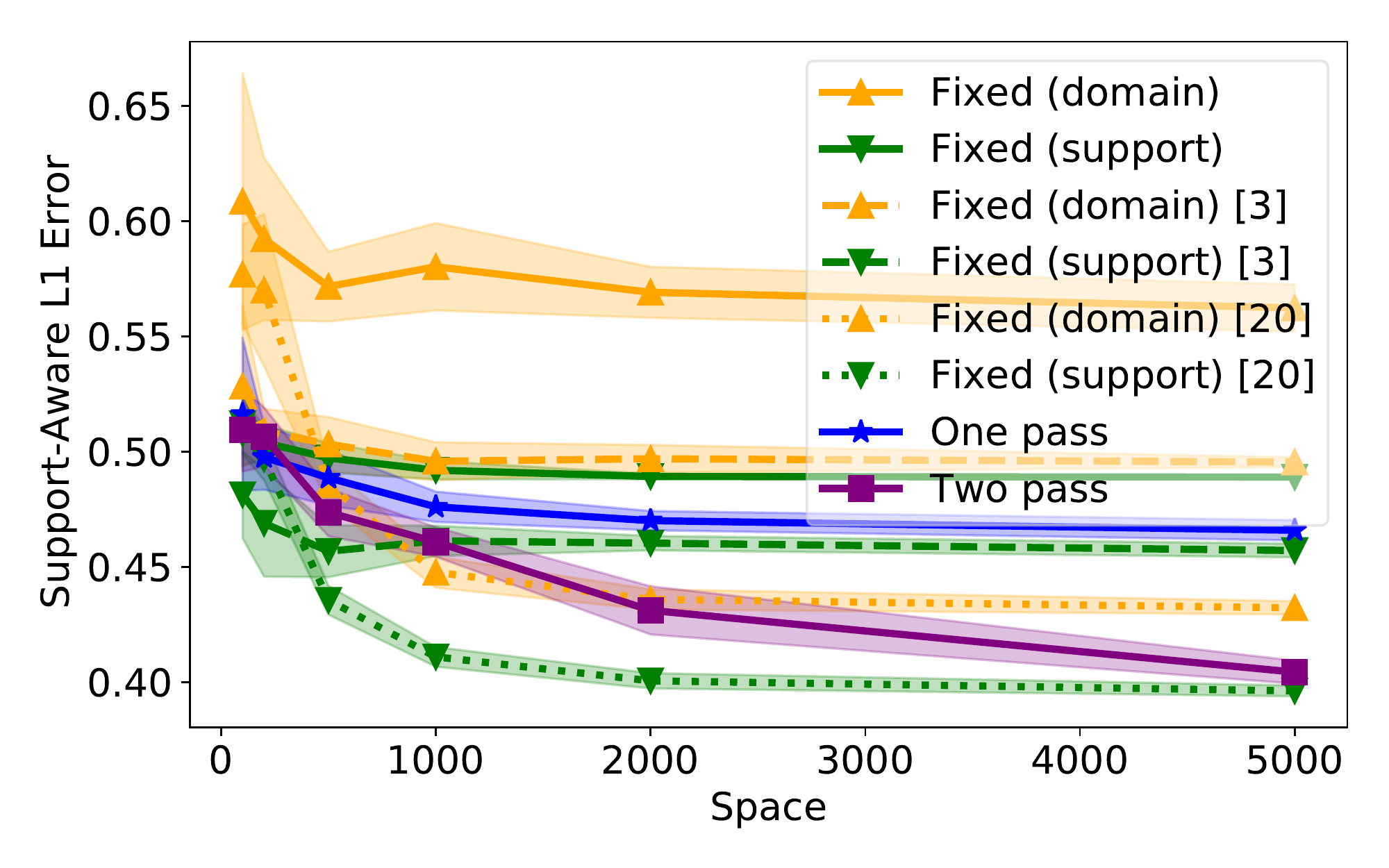}
\vspace{-1em}
\caption{Comparison of support-aware $L1$ error with varying space usage  with various $k$ on the Taxi dataset. Shading indicates one standard deviation over $10$ trials.}
\label{fig-taxivary}
\end{center}
\end{figure}

\begin{figure}[H]
\begin{center}
\includegraphics[width=0.4\columnwidth]{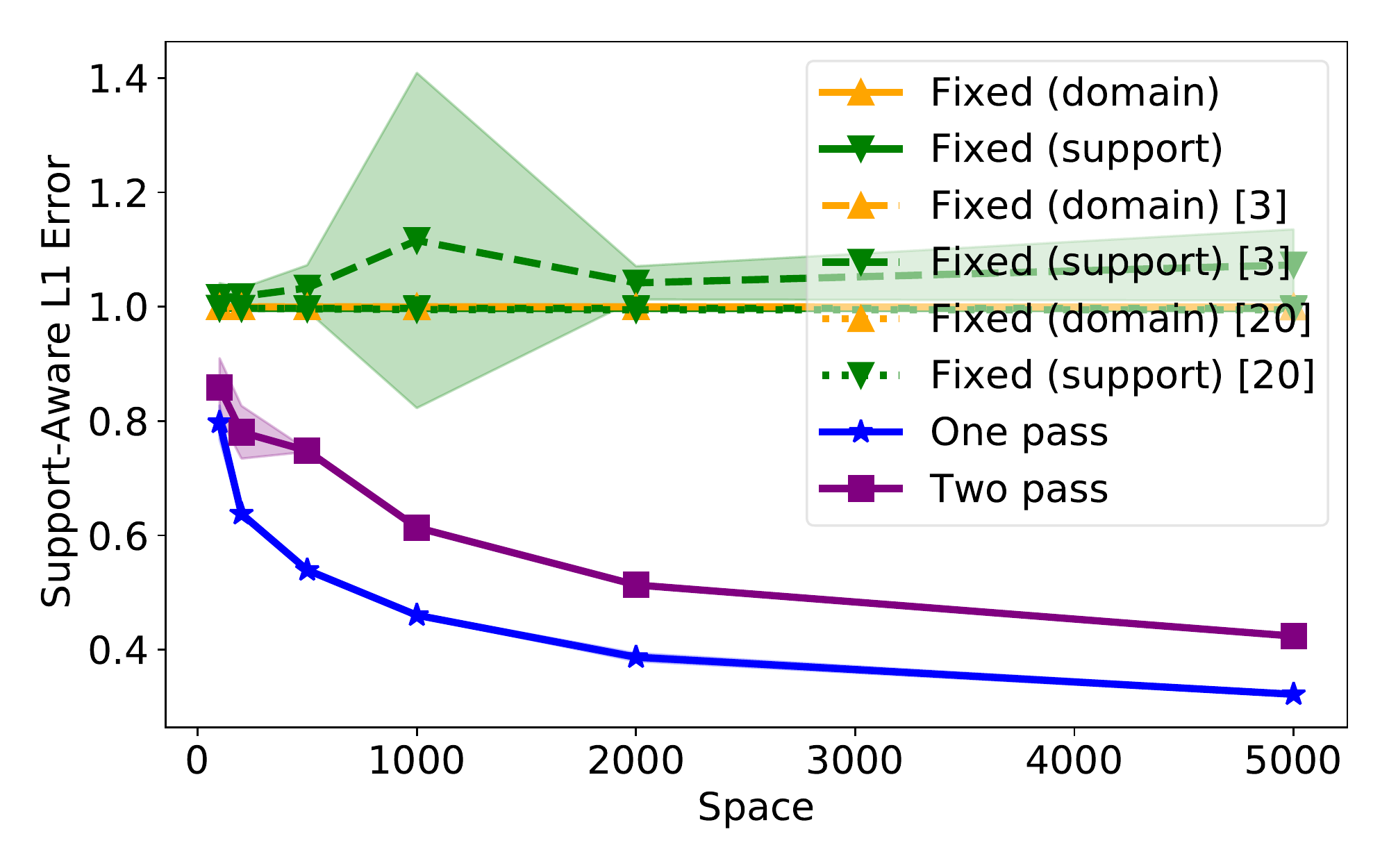}
\vspace{-1em}
\caption{Comparison of support-aware $L1$ error with varying space usage  with various $k$ on the CAIDA dataset. Shading indicates one standard deviation over $10$ trials.}
\label{fig-caida3-vary}
\end{center}
\end{figure}

\begin{figure}[H]
\begin{center}
\includegraphics[width=0.4\columnwidth]{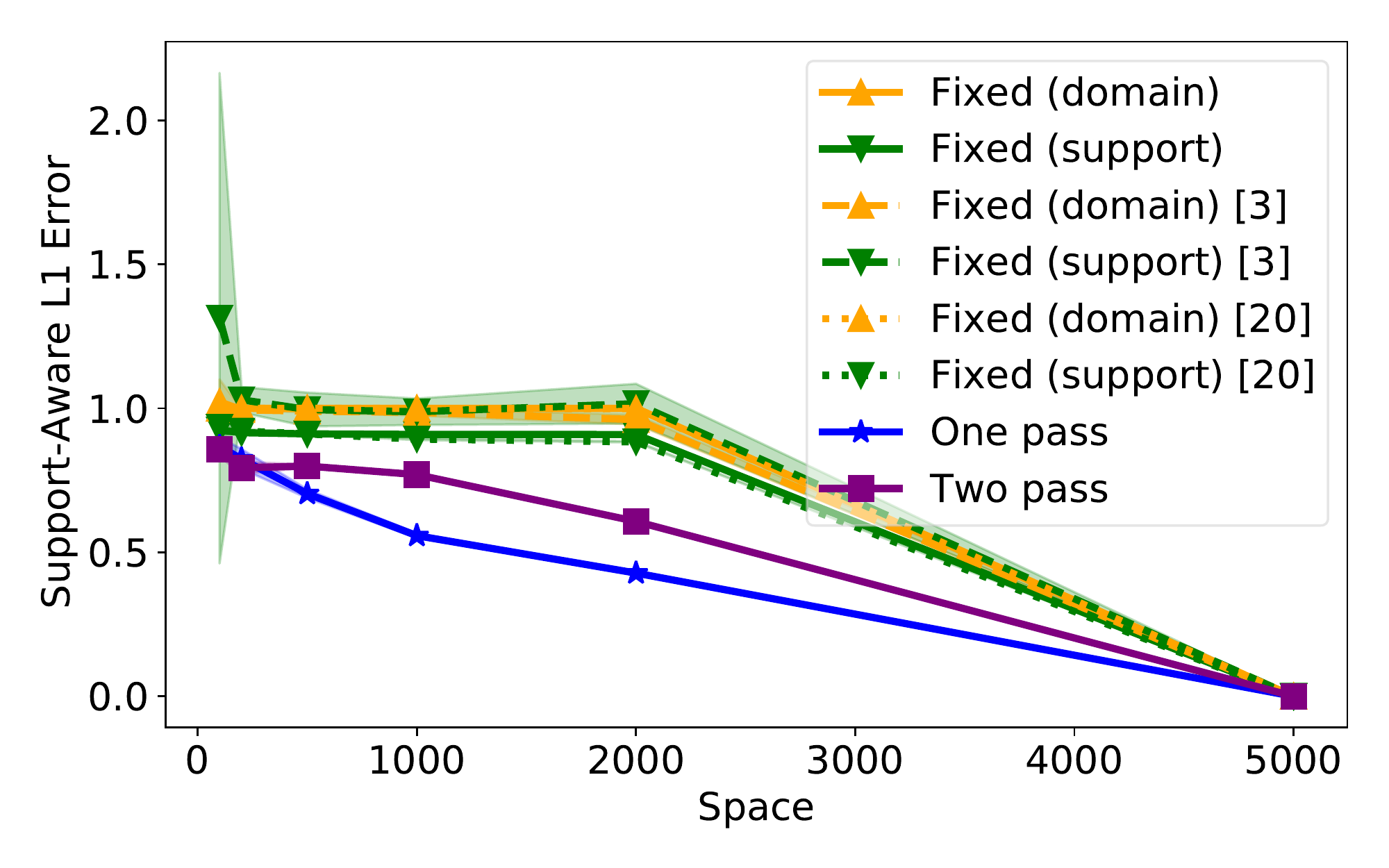}
\vspace{-1em}
\caption{Comparison of support-aware $L1$ error with varying space usage with various $k$ on the War and Peace dataset. Shading indicates one standard deviation over $10$ trials.}
\label{fig-warpeacevary}
\end{center}
\end{figure}

\begin{figure}[H]
\begin{center}
\includegraphics[width=0.4\columnwidth]{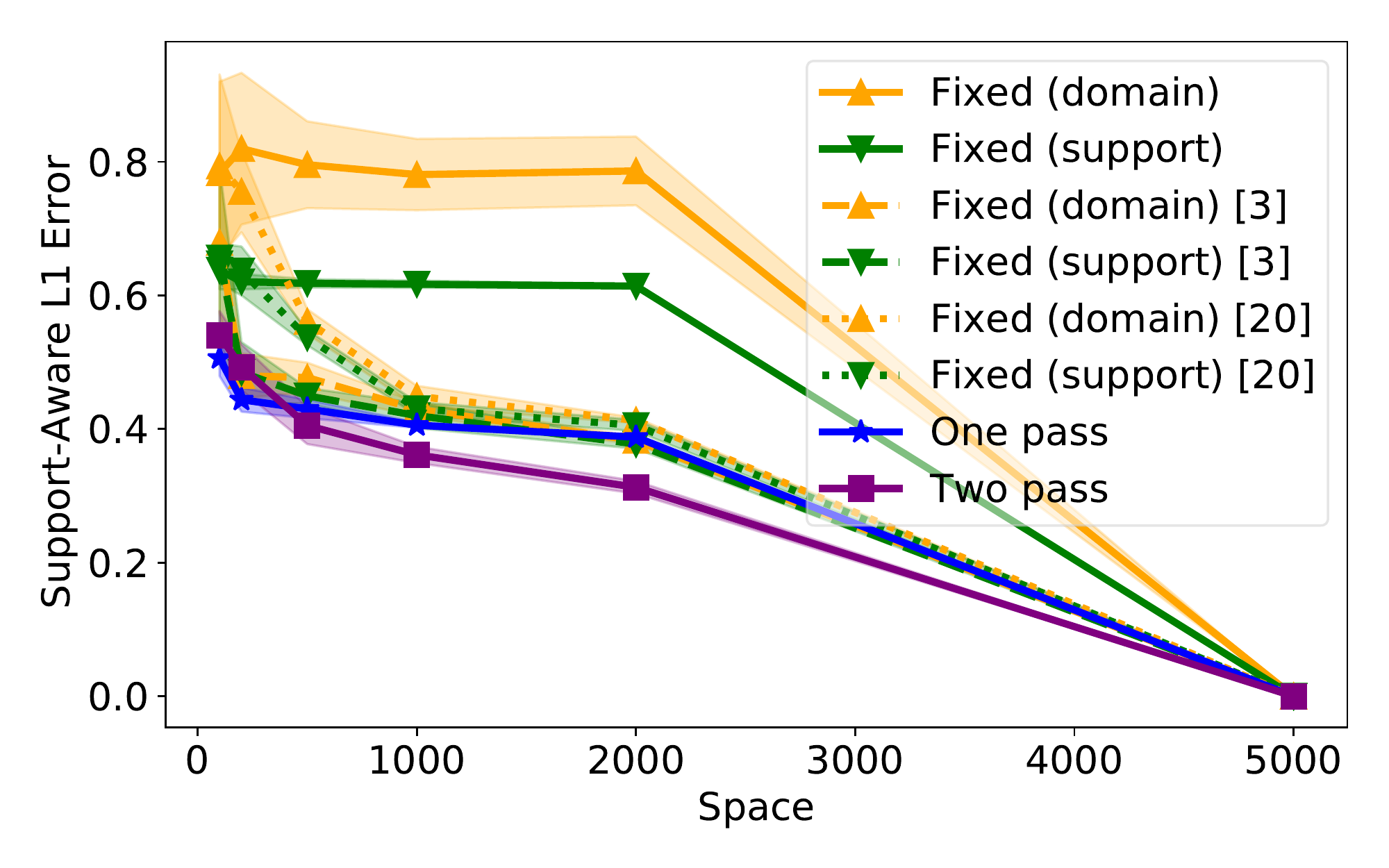}
\vspace{-1em}
\caption{Comparison of support-aware $L1$ error with varying space usage with various $k$ on the McDonalds dataset. Shading indicates one standard deviation over $10$ trials.}
\label{fig-mcdonaldsvary}
\end{center}
\end{figure}

\section{Additional Experiments: $k=10$}
\label{appendix:experiments-k10}
In this section, we display additional experimental results. In Section~\ref{sec:experiments}, we compare our algorithms against several baselines with $k=5$. Here, we present the same experiments with $k=10$.
There is qualitatively little difference between the performance of any of the algorithms or baselines with $k=5$ compared to $k=10$.

\begin{figure}[h]
\begin{center}
\includegraphics[width=0.4\columnwidth]{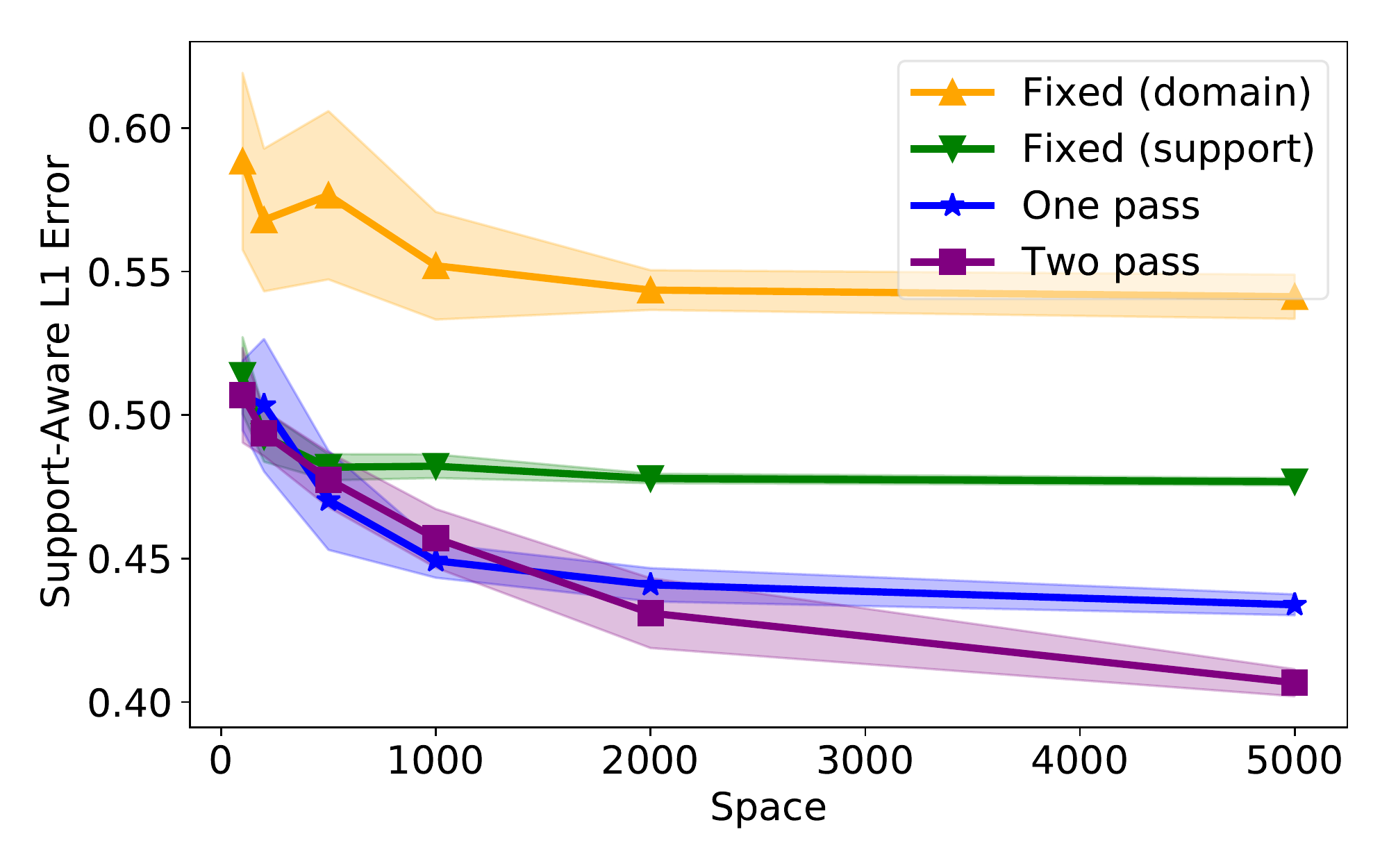}
\vspace{-1em}
\caption{Comparison of support-aware $L1$ error with varying space usage  with $k = 5$ on the Taxi dataset. Shading indicates one standard deviation over $10$ trials.}
\label{fig-taxi10}
\end{center}
\end{figure}

\begin{figure}[h]
\begin{center}
\includegraphics[width=0.4\columnwidth]{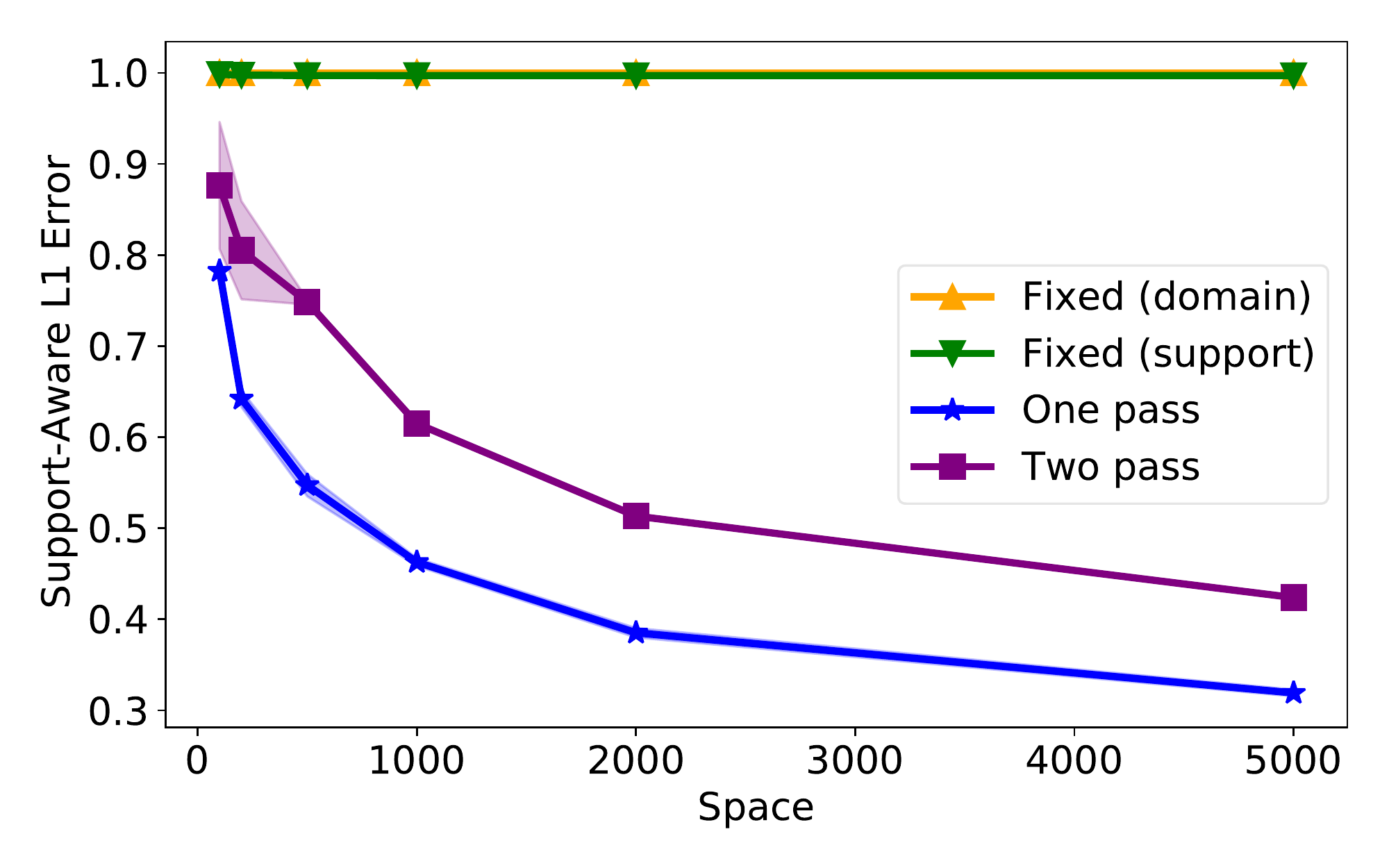}
\vspace{-1em}
\caption{Comparison of support-aware $L1$ error with varying space usage  with $k = 5$ on the CAIDA dataset. Shading indicates one standard deviation over $10$ trials.}
\label{fig-caida3-10}
\end{center}
\end{figure}

\begin{figure}[h]
\begin{center}
\includegraphics[width=0.4\columnwidth]{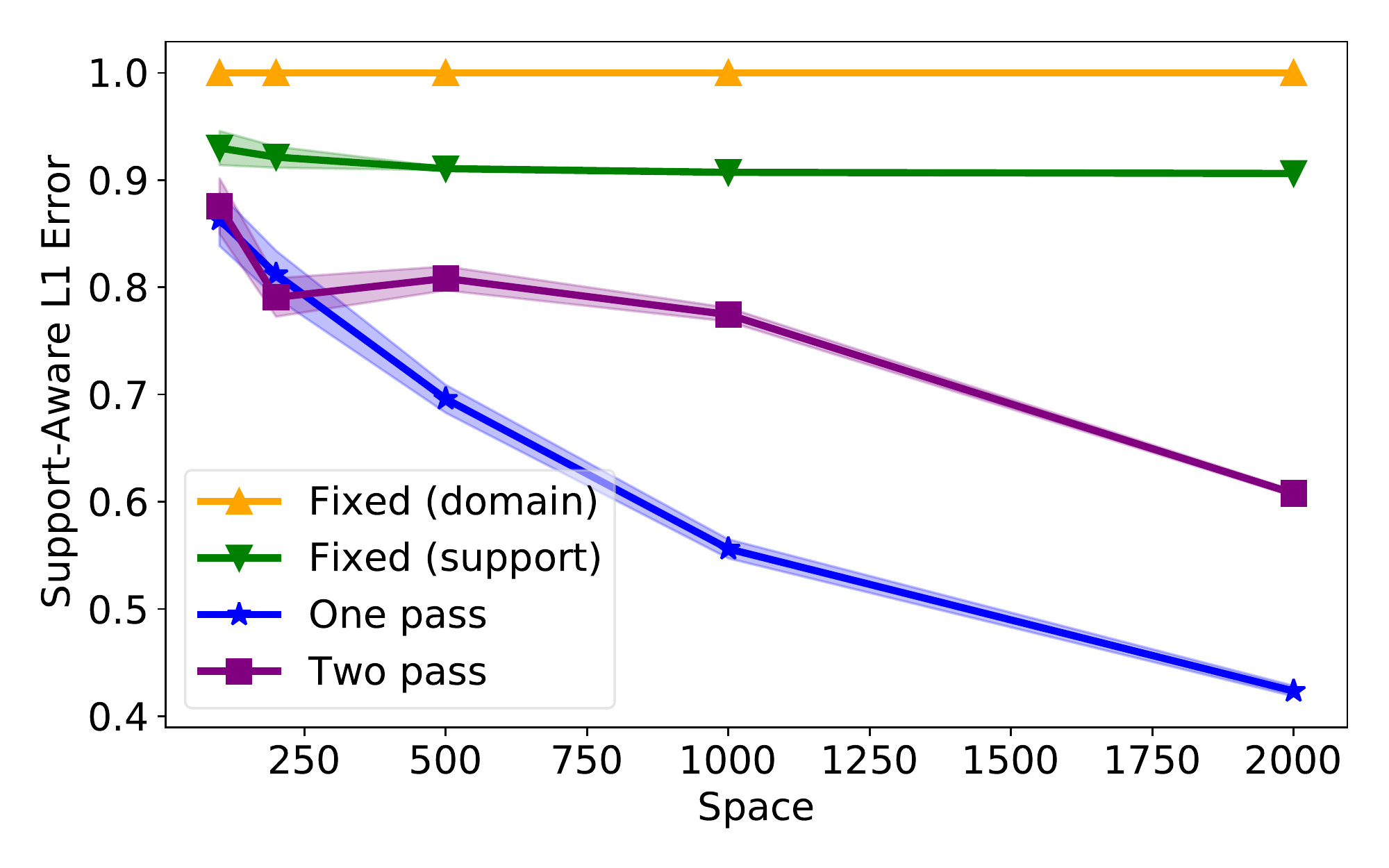}
\vspace{-1em}
\caption{Comparison of support-aware $L1$ error with varying space usage with $k = 5$ on the War and Peace dataset. Shading indicates one standard deviation over $10$ trials.}
\label{fig-warpeace10}
\end{center}
\end{figure}

\begin{figure}[h]
\begin{center}
\includegraphics[width=0.4\columnwidth]{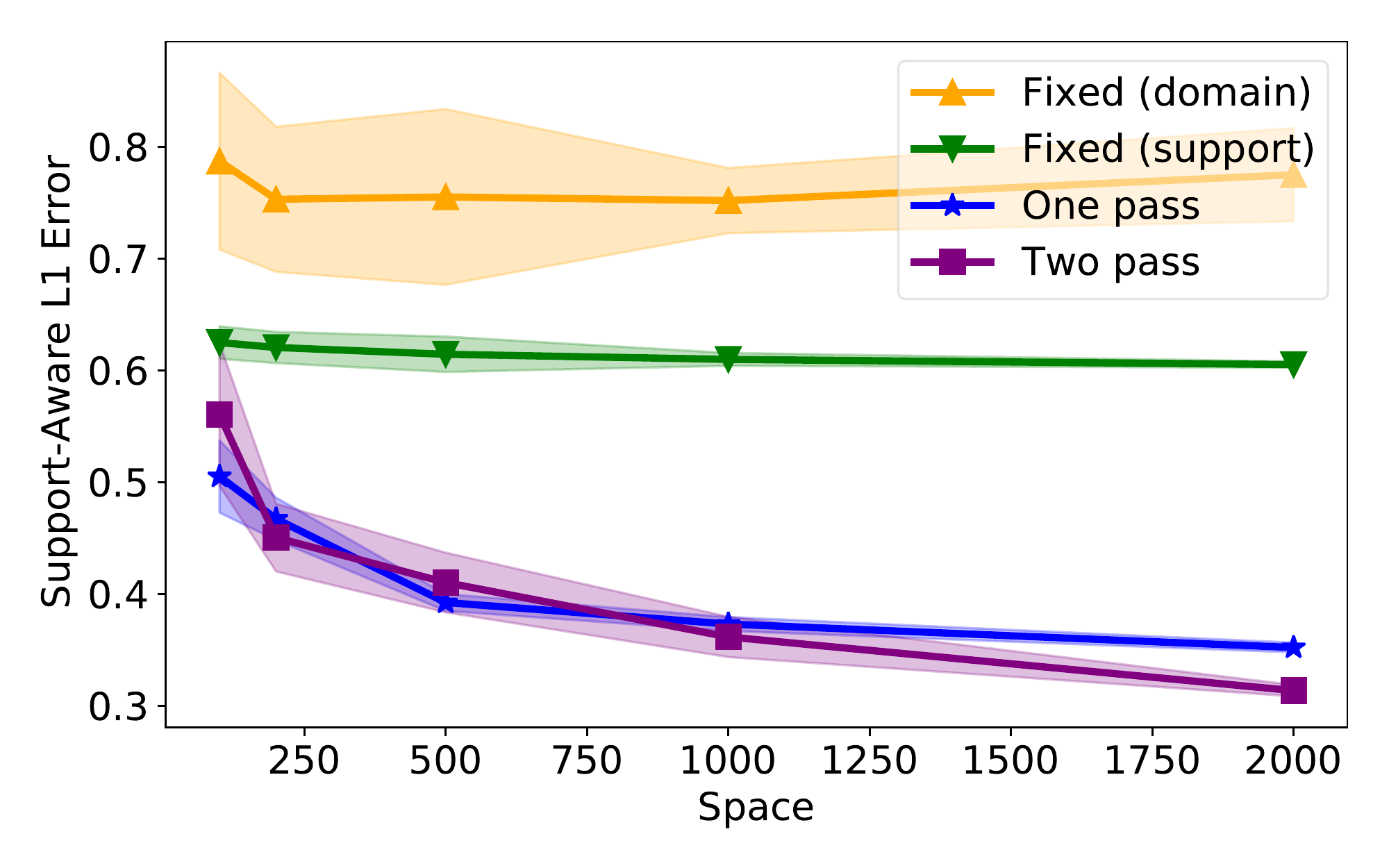}
\vspace{-1em}
\caption{Comparison of support-aware $L1$ error with varying space usage with $k = 5$ on the McDonalds dataset. Shading indicates one standard deviation over $10$ trials.}
\label{fig-mcdonalds10}
\end{center}
\end{figure}

\end{document}